\documentclass[a4paper,10pt]{article}
\usepackage{cite} 
\usepackage{enumerate} 
\usepackage{xspace}
\usepackage{authblk}
\usepackage[utf8]{inputenc}
\usepackage[margin=1.2in]{geometry}
\usepackage{amssymb}
\usepackage{mathrsfs}
\usepackage{euscript}
\usepackage{amsmath}
\usepackage{amsthm}
\usepackage{url}
\usepackage{hyperref}
\long\def\shortlong#1#2{#2}

\usepackage{tikz}
\tikzstyle{every picture} = [>=latex]
\usetikzlibrary{calc,arrows,decorations.markings}
\def\centerarc[#1] (#2) (#3:#4:#5)
    { \draw[#1] ($(#2)+({#5*cos(#3)},{#5*sin(#3)})$) arc (#3:#4:#5); }

\def\ca#1{{\cal#1}}
\def\scri#1{{\EuScript#1}}

\newcommand\sem{\setminus}
 
\newcommand\FO{FO\xspace}
\newcommand\EFO{$\exists$FO\xspace}

\newtheorem{thm} {Theorem}[section]
\newtheorem{lem} [thm]{Lemma}
\newtheorem{cor}[thm] {Corollary}

\newtheorem{prop}[thm] {Proposition}
\newtheorem{claim} [thm]{Claim}

\shortlong{\def\apxmark{{$\!\!$\bf*}\xspace}}{\def\apxmark{}}

\title{FO model checking of geometric graphs%
\footnote{Short version appeared at IPEC 2017.}}
\author[1]{Petr Hlin\v en\'y%
\thanks{P.~Hlin\v en\'y and F.Pokr\'yvka are supported by the 
Czech Science Foundation project No.~17-00837S.}%
}
\author[1]{Filip Pokr\'yvka$^*$}
\author[1]{Bodhayan Roy}
\affil[1]{Faculty of Informatics, Masaryk University Brno, Czech Republic\\
	\texttt{\{hlineny,\,xpokryvk,\,b.roy\}@fi.muni.cz}}

\begin{document}

\maketitle

\begin{abstract}
Over the past two decades the main focus of research into first-order (FO)
model checking algorithms has been on sparse relational 
structures -- culminating in the FPT algorithm by Grohe, Kreutzer and Siebertz for
FO model checking of nowhere dense classes of graphs.
On contrary to that, except the case of locally bounded clique-width
only little is currently known about FO model checking of dense classes 
of graphs or other structures.
We study the FO model checking problem for dense graph classes
definable by geometric means (intersection and visibility graphs).
We obtain new nontrivial FPT results, e.g., for restricted subclasses of
{\em circular-arc, circle, box, disk, and polygon-visibility graphs}.
These results use the FPT algorithm by Gajarsk\'y~et~al.\
for FO model checking of posets of bounded width.
We also complement the tractability results by related hardness reductions.
\medskip

\end{abstract}

\section{Introduction}
\label{sec:intro}

Algorithmic meta-theorems are results stating that all problems expressible
in a certain language are efficiently solvable on certain classes of
structures, e.g.\ of finite graphs. 
Note that the model checking problem for {\em first-order logic} -- 
given a graph $G$ and an \FO formula $\phi$, we want to decide whether $G$
satisfies $\phi$ (written as $G\models\phi$) -- is trivially solvable
in time $|V(G)|^{\mathcal{O}(|\phi|)}$.
``Efficient solvability'' hence in this context often means
{\em fixed-parameter tractability} (FPT);
that is, solvability in time $f(|\phi|)\cdot|V(G)|^{\mathcal{O}(1)}$ 
for some computable function~$f$.  

In the past two decades
algorithmic meta-theorems for \FO logic on sparse graph classes
received considerable attention.  
While the algorithm of~\cite{cmr00} for MSO on graphs of bounded clique-width
implies fixed-parameter tractability of \FO model checking on graphs of
locally bounded clique-width via Gaifman's locality,
one could go far beyond that.
After the result of Seese~\cite{Seese96}
proving fixed-parameter tractability of \FO model checking on graphs of
bounded degree there followed a series of results~\cite{FrickG01, DawarGK07,
  DvorakKT10} establishing the same conclusion for increasingly rich sparse
graph classes. This line of research culminated in the result of Grohe,
Kreutzer and Siebertz~\cite{gks14}, who proved that \FO model checking is FPT
on {\em nowhere dense} graph classes.

While the result of~\cite{gks14} is the best possible in the following sense---%
if a graph class $\ca D$ is {\em monotone} (closed on taking subgraphs)
and not nowhere dense, then 
the \FO model checking problem on $\ca D$ is as hard as that on all graphs; 
this does not exclude interesting FPT meta-theorems on {\em somewhere dense}
non-monotone graph classes.
Probably the first extensive work of the latter dense kind,
beyond locally bounded clique-width, was that of 
Ganian et al.~\cite{GHKOST15} studying subclasses of interval graphs 
in which \FO model checking is FPT (precisely, those which use only 
a finite set of interval lengths).
Another approach has been taken in the works of Bova, Ganian and
Szeider~\cite{bgs16} and Gajarsk\'y et al.~\cite{gajarskyetal15},
which studied \FO model checking on posets --
posets can be seen as typically quite dense special digraphs.
Altogether, however, only very little is known about \FO model checking
of somewhere dense graph classes (except perhaps specialised \cite{GHOLR16}).

The result of Gajarsk\'y et al.~\cite{gajarskyetal15} 
claims that \FO model checking is FPT on posets of bounded width 
(size of a maximum antichain), 
and it happens to imply \cite{GHKOST15} in a stronger setting (see below).
One remarkable message of~\cite{gajarskyetal15} is the following (citation):
{\em The result may also be used directly towards establishing
fixed-parameter tractability for \FO model checking of other graph classes. 
Given the ease with which it (\cite{gajarskyetal15}\,)
implies the otherwise non-trivial result on interval graphs \cite{GHKOST15}, 
it is natural to ask what other (dense) graph classes can be 
interpreted in posets of bounded width.}
Inspired by the geometric case of interval graphs, we propose to study dense graph 
classes defined in geometric terms, such as intersection and visibility graphs,
with respect to tractability of their \FO model checking problem.

The motivation for such study is a two-fold.
First, intersection and visibility graphs
present natural examples of non-monotone somewhere dense graph classes
to which the great ``sparse'' \FO tractability result of~\cite{gks14} cannot be 
(at least not easily) applied.
Second, their supplementary geometric structure allows to better
understand (as we have seen already in~\cite{GHKOST15})
the boundaries of tractability of \FO model checking on them, which is,
to current knowledge, terra incognita for hereditary graph classes in
general.

\medskip
Our results mainly concern graph classes which are related to interval graphs.
Namely, we prove (Theorem~\ref{thm:CAtractable}) that
\FO model checking is FPT on {\em circular-arc graphs} (these are interval graphs
on a circle) if there is no long chain of arcs nested by inclusion.
This directly extends the result of~\cite{GHKOST15} and its aforementioned 
strengthening in~\cite{gajarskyetal15}
(with bounding chains of nested intervals instead of their lengths).
We similarly show tractability of \FO model checking of interval-overlap graphs, 
also known as {\em circle graphs}, of bounded independent set size 
(Theorem~\ref{thm:CIRitractable}),
and of restricted subclasses of {\em box and disk graphs} 
which naturally generalize interval graphs 
to two dimensions (Theorem~\ref{thm:Boxtractable}%
\shortlong{}{ and~\ref{thm:Disktractable}}).

On the other hand, 
for all of the studied cases we also show that whenever we relax our
additional restrictions (parameters), 
the \FO model checking problem becomes as hard on our intersection 
classes as on all graphs (Corollary~\ref{cor:allhardness}).
Some of our hardness claims hold also for the weaker \EFO model checking problem
(Proposition~\ref{pro:EFOhardness}).

Another well studied dense graph class in computational geometry
are {\em visibility graphs} of polygons, which
have been largely explored in the context of recognition, partition, 
guarding and other optimization problems 
\cite{g-vap-07,o-agta-87}.
We consider some established special cases,
involving {\em weak visibility, terrain} and {\em fan} polygons.
We prove that \FO model checking is FPT for the visibility graphs of 
a weak visibility polygon of a convex edge, with bounded number of reflex
(non-convex) vertices (Theorem~\ref{thm:VIStractable}).
On the other hand, without bounding reflex vertices,
\FO model checking remains hard even for the much more special case
of polygons that are terrain and convex fans at the same time
(Theorem~\ref{thm:TerFan-hard}).

As noted above, our fixed-parameter tractability proofs use the strong
result~\cite{gajarskyetal15} on \FO model checking of posets of bounded width.
We refer to Section~\ref{sec:prelim} for a detailed explanation of the
technical terms used here.
Briefly, for a given graph $G$ from the respective class and a formula
$\phi$, we show how to efficiently construct a poset $\ca P_G$ of bounded width
and a related \FO formula $\phi^I$ such that $G\models\phi$
iff $\ca P_G\models \phi^I$, and then solve the latter problem.
\shortlong{}{%
In constructing the poset $\ca P_G$ we closely exploit the respective 
geometric representation of~$G$.
}

With respect to the previously known results, we remark that our graph
classes are not sparse, as they all contain large complete or complete
bipartite subgraphs.
For many of them, namely unit circular-arc graphs, circle graphs
of bounded independence number, and unit box and disk graphs,
we can also show that they are of locally unbounded clique-width
by a straightforward adaptation of an argument from~\cite{GHKOST15}
(Proposition~\ref{prop:unboundedcw}).
For the visibility graphs of a weak visibility polygon of a convex edge,
we leave the question of bounding their local clique-width open.

\medskip
Lastly, we particularly emphasize the seemingly simple tractable case
(Corollary~\ref{cor:PERMtractable}) of permutation graphs of bounded clique size:
in relation to so-called stability notion (cf.~\cite{DBLP:journals/ejc/AdlerA14}),
already the hereditary class of triangle-free permutation graphs
has the $n$-order property (i.e., is {\em not} stable),
and yet \FO model checking of this class is FPT.
This example presents a natural hereditary and non-stable graph class
with FPT \FO model checking other than, say, graphs of bounded clique-width.
We suggest that if we could fully understand the precise breaking point(s)
of FP tractability of \FO model checking on simply described intersection
classes like the permutation graphs, then we would get much better insight
into FP tractability of \FO model checking of general hereditary graph classes.

\shortlong{%
Due to space restrictions, most of the proofs and some illustrating
pictures have had to be removed from this short paper.
The statements with removed proofs are marked by~~\apxmark
and they can be found, for example, in the arXiv version.
}{}%

\section{Preliminaries}\label{sec:prelim}

\shortlong{}{%
We recall some established concepts concerning intersection graphs
and first-order logic.
}
\subparagraph{Graphs and intersection graphs.} 
We work with {\em finite simple undirected graphs}
and use standard graph theoretic notation.
We refer to the vertex set of a graph $G$ as to $V(G)$
and to its edge set as to~$E(G)$, and we write shortly $uv$ for an
edge~$\{u,v\}$.
As it is common in the context of FO logic on graphs,
vertices of our graphs can carry arbitrary labels.

Considering a family of sets $\scri S$ (in our case, of geometric objects in
the plane), the {\em intersection graph of~$\scri S$}
is the simple graph $G$ defined by $V(G):=\scri S$ and
$E(G):=\{AB: A,B\in\scri S,\, A\cap B\not=\emptyset\}$.
In respect of algorithmic questions,
it is important to distinguish whether an intersection graph $G$ is given
on the input as an abstract graph $G$, or alongside with its intersection
representation~$\scri S$.
\shortlong{}{%
Usually, finding an appropriate representation for given $G$ is a hard task,
but we will mostly restrict our attention to intersection classes for which
there exists a polynomial-time algorithm for computing the representation.
}

One folklore example of a widely studied intersection graph class
are {\em interval graphs} -- the intersection graphs of intervals on the
real line.
Interval graphs enjoy many nice algorithmic properties, e.g.,
their representation can be constructed quickly,
and generally hard problems like clique, independent set and chromatic
number are solvable in polynomial time for them.

For a general overview and extensive reference guide of intersection graph
classes we suggest to consult the online system ISGCI~\cite{isgci}.
\shortlong{}{%
Regarding {\em visibility graphs}, which present a kind of geometric graphs
behaving very differently from intersection graphs, we refer to
Section~\ref{sec:visibility} for their separate more detailed treatment.
}

\subparagraph{FO logic.}
The {\em first-order logic of graphs} (abbreviated as \FO) applies the
standard language of first-order logic to a graph $G$ viewed as a relational 
structure with the domain $V(G)$ and the single binary (symmetric) relation $E(G)$.
That is, in graph \FO we have got the standard predicate $x=y$, 
a binary predicate $edge(x,y)$ with the usual meaning $xy\in E(G)$, 
an arbitrary number of unary predicates $L(x)$ with the meaning that $x$
holds the label~$L$,
usual logical connectives $\wedge,\vee,\to$, and quantifiers
$\forall x$, $\exists x$ over the vertex set $V(G)$.

For example, $\phi(x,y)\equiv \exists z\big(edge(x,z)\wedge edge(y,z)
	\wedge red(z)\big)$
states that the vertices $x,y$ have a common neighbour in~$G$
which has got label `red'.
One can straightforwardly express in \FO properties such as $k$-clique
$\exists x_1,\dots,x_k\big(\bigwedge_{i<j=1}^k
	(edge(x_i,x_j)\wedge x_i\not=x_j)\big)$
and $k$-dominating set
$\exists x_1,\dots,x_k\forall y\big(\bigvee_{i=1}^k
	(edge(x_i,y)\vee y=x_i)\big)$.
Specially, an \FO formula $\phi$ is {\em existential} (abbreviated as \EFO)
if it can be written as $\phi\equiv\exists x_1,\dots,x_k\,\psi$
where $\psi$ is quantifier-free.
For example, $k$-clique is \EFO while $k$-dominating set is not.

Likewise, \FO logic of posets treats a poset $\ca P=(P,\sqsubseteq)$ 
as a finite relational structure with the domain $P$ 
and the (antisymmetric) binary predicate~$x\sqsubseteq y$
(instead of the predicate $edge$) with the usual meaning.
Again, posets can be arbitrarily labelled by unary predicates.

\subparagraph{Parameterized model checking.}
Instances of a parameterized problem can be considered as pairs
$\langle I,k\rangle$ where $I$ is the {main part} of the instance and $k$ is
the \emph{parameter} of the instance; the latter is usually a
non-negative integer.  A parameterized problem is
\emph{fixed-parameter tractable (FPT)} if instances $\langle I,k\rangle$ 
of size $n$ can be solved in time $O(f(k)\cdot n^c)$ where $f$ is a 
computable function and $c$ is a constant independent of $k$.
In {\em parameterized model checking}, instances are
considered in the form $\langle(G,\phi),|\phi|\rangle$
where $G$ is a structure, $\phi$ a formula, the question is whether
$G\models\phi$ and the parameter is the size of~$\phi$.

When speaking about the \FO model checking problem in this paper, 
we always implicitly consider the formula $\phi$ (precisely its size) as a parameter.
We shall use the following result:
\begin{thm}[$\!$\cite{gajarskyetal15}]
\label{thm:posetFPT}
The \FO model checking problem of (arbitrarily labelled) posets, i.e., 
deciding whether $\ca P\models\phi$ for a labelled poset $\ca P$ and \FO $\phi$,
is fixed-parameter tractable with respect to $|\phi|$ and the width of 
$\ca P$ (this is the size of the largest antichain in~$\ca P$).
\end{thm}

We also present, for further illustration, a result on \FO model checking
of interval graphs with bounded nesting.
A set $\ca A$ of intervals (interval representation) is called {\em proper}
if there is no pair of intervals in $\ca A$ such that one is contained
in the other.
We call $\ca A$ a {\em$k$-fold proper} set of intervals
if there exists a partition ${\ca A}={\ca A}_1\cup\dots\cup {\ca A}_k$
such that each ${\ca A}_j$ is a proper interval set for $j=1,\dots,k$.
Clearly, $\ca A$ is $k$-fold proper if and only if
there is no chain of $k+1$ inclusion-nested intervals in~$\ca A$.
From Theorem~\ref{thm:posetFPT} one can, with help of relatively easy
arguments (Lemma~\ref{INTtoposet}), derive the following:

\begin{thm}[$\!$\cite{gajarskyetal15},
cf.~Proposition~\ref{prop:reducetoposet} and Lemma~\ref{INTtoposet}]
\label{thm:INTtractable}
Let $G$ be an interval graph given alongside with its 
$k$-fold proper interval representation $\ca A$.
Then FO model checking of $G$ is FPT with respect to the parameters 
$k$ and the formula size.
\end{thm}

\subparagraph{Parameterized hardness.}
For some parameterized problems, like the $k$-clique on all graphs,
we do not have nor expect any FPT algorithm.
To this end, the theory of parameterized complexity of Downey and 
Fellows~\cite{DBLP:series/txcs/DowneyF13} defines complexity
classes $W[t]$, $t\geq 1$, such that the $k$-clique problem is
complete for $W[1]$ (the least class).
Furthermore, theory also defines a larger complexity class $AW[*]$ 
containing all of $W[t]$.
Problems that are $W[1]$-hard do not admit an FPT algorithm unless the
established Exponential Time Hypothesis fails.

\begin{thm}[$\!$\cite{DBLP:conf/dmtcs/DowneyFT96}]
\label{thm:FOgenhard}
The \FO model checking problem 
(where the formula size is the parameter)
of all simple graphs is $AW[*]$-complete.
\end{thm}

Dealing with parameterized hardness of \FO model checking, one should also
mention the related {\em induced subgraph isomorphism problem}:
for a given input graph $G$, and a graph $H$ as the parameter,
decide whether $G$ has an induced subgraph isomorphic to~$H$.
Note that this includes the clique and independent set problems.
Induced subgraph isomorphism (parameterized by the subgraph size)
is clearly a weaker problem than parameterized \FO model checking,
since one may ``guess'' the subgraph with $|V(H)|$ existential quantifiers 
and then verify it edge by edge.
Consequently, every parameterized hardness result for induced subgraph
isomorphism readily implies same hardness results for \EFO and \FO model
checking.

\subparagraph{\FO interpretations.}
Interpretations are a standard tool of logic and finite model theory.
To keep our paper short, we present here only a simplified description of
them, tailored specifically to our need of interpreting geometric graphs in
posets.

An {\em\FO interpretation} is a pair $I=(\nu,\psi)$ of poset \FO formulas
$\nu(x)$ and $\psi(x,y)$ (of one and two free variables, respectively).
For a poset $\ca P$, this defines a graph $G:=I(\ca P)$ such that
$V(G)=\{v: \ca P\models\nu(v)\}$ and
$E(G)=\{uv: u,v\in V(G),\, \ca P\models\psi(u,v)\vee\psi(v,u)\}$.
Possible labels of the elements are naturally inherited from $\ca P$ to~$G$.
Moreover, for a graph \FO formula $\phi$ the interpretation $I$
defines a poset \FO formula $\phi^I$ recursively as follows:
every occurrence of ${edge}(x,y)$ is replaced by $\psi(x,y)\vee\psi(y,x)$,
every $\exists x\,\sigma$ is replaced by $\exists x\,(\nu(x)\wedge\sigma)$
and $\forall x\,\sigma$ by $\forall x\,(\nu(x)\to\sigma)$.
Then, obviously, $\ca P\models\phi^I \Longleftrightarrow G\models\phi$.

Usefulness of the concept is illustrated by the following trivial claim:
\begin{prop}\apxmark\label{prop:reducetoposet}
Let $\scri P$ be a class of posets such that the \FO model checking problem
of $\scri P$ is FPT, and let $\scri G$ be a class of graphs.
Assume there is a computable \FO interpretation~$I$, and
for every graph $G\in\scri G$ we can in polynomial time compute
a poset $\ca P\in\scri P$ such that $G=I(\ca P)$.
Then the \FO model checking problem of $\scri G$ is in FPT.
\end{prop}
\shortlong{}{\begin{proof}
Given $G\in\scri G$ and formula $\phi$ (the parameter),
we construct $\phi^I$ and $\ca P\in\scri P$ such that $G=I(\ca P)$,
and call the assumed algorithm to decide $\ca P\models\phi^I$.
\end{proof}}

\section{Tractability for Intersection Classes}\label{sec:tractability}

\subsection{Circular-arc graphs}\label{sec:CA}

{\em Circular-arc graphs} are intersection graphs of arcs (curved intervals)
on a circle.
They clearly form a superclass of interval graphs, and they enjoy similar
nice algorithmic properties as interval graphs,
such as efficient construction of the representation~\cite{zbMATH02055149},
and easy computation of, say, maximum independent set or clique.

Since the \FO model checking problem is $AW[*]$-complete on interval
graphs~\cite{GHKOST15},
the same holds for circular-arc graphs in general.
Furthermore, by~\cite{DBLP:journals/algorithmica/MarxS13,HEGGERNES2015252}
already \EFO model checking is $W[1]$-hard for interval and circular-arc graphs.
A common feature of these hardness reductions 
(see more discussion in Section~\ref{sec:hardness})
is their use of unlimited chains of nested intervals/arcs.
Analogously to Theorem~\ref{thm:INTtractable},
we prove that considering only {\em$k$-fold proper circular-arc} representations
(the definition is the same as for $k$-fold proper interval representations)
makes \FO model checking of circular-arc graphs tractable.

\begin{thm}
\label{thm:CAtractable}
Let $G$ be a circular-arc graph given alongside with its 
$k$-fold proper circular-arc representation $\ca A$.
Then \FO model checking of $G$ is FPT with respect to the parameters 
$k$ and the formula size.
\end{thm}

Note that we can (at least partially) avoid the assumption of having a
representation $\ca A$ in the following sense.
Given an input graph $G$, we compute a circular-arc representation $\ca A$
using~\cite{zbMATH02055149}, and then we easily determine the least $k'$ such
that $\ca A$ is $k'$-fold proper.
However, without further considerations, this is not guaranteed to provide
the minimum $k$ over all circular-arc representations of~$G$,
and not even $k'$ bounded in terms of the minimum $k$.

Our proof will be based on the following extension of the related argument
from~\cite{gajarskyetal15}:
\begin{lem}[parts from {\cite[Section~5]{gajarskyetal15}}]\apxmark
\label{INTtoposet}
Let $\ca B$ be a $k$-fold proper set of intervals for some integer~$k>0$,
such that no two intervals of $\ca B$ share an endpoint.
There exist formulas $\nu,\psi,\vartheta$ depending on~$k$,
and a labelled poset $\ca P$ of width $k+1$ computable in polynomial time 
from~$\ca B$, such that all the following hold:
\begin{itemize}
\item The domain of $\ca P$ includes (the intervals from)~$\ca B$, 
	and $\ca P\models\nu(x)$ iff $x\in\ca B$,
\item $\ca P\models\psi(x,y)$ for intervals $x,y\in\ca B$ iff 
	$x\cap y\not=\emptyset$ (edge relation of the interval graph of~$\ca B$),
\item $\ca P\models\vartheta(x,y)$ for intervals $x,y\in\ca B$ iff 
	$x\subseteq y$ (containment of intervals).
\end{itemize}
\end{lem}
\shortlong{}{%
\def\PPP{\ca P}
\begin{proof}
The first part repeats an argument from \cite[Section~5]{gajarskyetal15}.
Let $D:=\{a,b: [a,b]\in\ca B\}$ be the set of all interval ends, and
$\ca B=\ca B_1\cup\dots\cup \ca B_k$ be such that each 
$\ca B_j$ is a proper interval set for $j=1,\dots,k$.
Let $P:=D\cup\ca B$.
We define a poset $\PPP=(P,\leq^\PPP)$ as follows:
\begin{itemize}
\item
for numbers $d_1,d_2\in D$ it is $d_1\leq^\PPP\!d_2$ iff $d_1\leq d_2$,
\item
for $j\in\{1,\dots,k\}$ and intervals $t_1,t_2\in\ca B_j$, it is
$t_1\leq^\PPP\!t_2$ iff $t_1$ is not to the right of $t_2$,
\item
for every $t=[a,b]\in\cal B$ and every $d\in D$, it is
$t\leq^\PPP\!d$ iff $d\geq b$, and $d\leq^\PPP\!t$ iff $d\leq a$.
\end{itemize}
An informal meaning of this definition of $\ca P$ is that
every interval $[a,b]$ from $\ca B$ is larger than its left end
$a$ (and hence larger than all interval ends before $a$),
and the interval is smaller than its right end $b$
(and hence smaller than all interval ends after $b$).
The interval $[a,b]$ is incomparable with all ends (of other intervals)
which are strictly between $a$ and $b$.

Using that each $\ca B_j$ is proper, one can verify that
$\ca P$ indeed is a poset.
The set $P$ can be partitioned into $k+1$ chains;
$D$ and $\ca B_1,\dots,\ca B_k$.
Hence the width of $\PPP$ is at most $k+1$.

In order to define the formulas, we give a special label `$D$'
to the set~$D$.
Then
\begin{align}
\nu(x) \equiv\>& \neg D(x),
\label{eq:nu}\\
\psi(x,y) \equiv\>&  \,\forall z\, \left[ D(z) \to \left(
		(\neg\, x\leq^\PPP\!z \vee \neg\, z\leq^\PPP\!y) \wedge
		(\neg\, y\leq^\PPP\!z \vee \neg\, z\leq^\PPP\!x)
	\right)\right],
\label{eq:psi}\\\label{eq:theta}
\vartheta(x,y) \equiv\>&  \,\forall z\, \left[ D(z) \to \left(
		(z\leq^\PPP\!y \to z\leq^\PPP\!x) \wedge
		(z\geq^\PPP\!y \to z\geq^\PPP\!x)
	\right)\right],
\end{align}
where the meaning of \eqref{eq:nu} is obvious,
\eqref{eq:psi} says that no interval end ($z$) is ``between''
the intervals $x,y$, and
\eqref{eq:theta} says that the left end of the interval $x$ is after that of
$y$ and the right end of $x$ is before (or equal) that of~$y$.
Consequently, $\ca P\models\nu(x)$ iff $x\in\ca B$,~
$\ca P\models\psi(x,y)$ iff none of the intervals $x,y$ is fully to the left
of the other (and so $x\cap y\not=\emptyset$),
and $\ca P\models\vartheta(x,y)$ iff $x\subseteq y$, as required.
\end{proof}
}

\begin{figure}[t]
$$
\begin{tikzpicture}[scale=\shortlong{2.7}{3.5}]
\tikzstyle{every path}=[draw,color=gray];
\draw[dashed] (0,0) circle (5mm);
\draw (0,-0.41) node {0};
\draw[dashed,thin] (0,-0.48) -- (0,-0.67) ;
\tikzstyle{every path}=[draw,color=black, |-|];
\centerarc[] (0,0) (5:85:0.55) ;
\centerarc[] (0,0) (-35:65:0.6) ;
\centerarc[] (0,0) (-75:20:0.65) ;
\centerarc[] (0,0) (45:125:0.65) ;
\centerarc[] (0,0) (70:170:0.6) ;
\centerarc[] (0,0) (110:220:0.55) ;
\centerarc[] (0,0) (210:245:0.65) ;
\centerarc[] (0,0) (190:235:0.6) ;
\tikzstyle{every path}=[draw,color=blue,thick, |-|];
\centerarc[] (0,0) (255:310:0.6) ;
\centerarc[] (0,0) (225:300:0.55) ;
\tikzstyle{every path}=[draw,color=red,thick,dashed, |-|];
\centerarc[] (0,0) (-50:255:0.45) ;
\centerarc[] (0,0) (-60:225:0.4) ;
\end{tikzpicture}
$$
\caption{An illustration; 
a proper circular-arc representation $\ca A$ (ordinary black and thick blue arcs),
giving raise to a $2$-fold proper interval set $\ca B$
(ordinary black and dashed red arcs), as in the proof of Theorem~\ref{thm:CAtractable}.
The red arcs are complements of the corresponding blue arcs.}
\label{fig:arc-complement}
\end{figure}
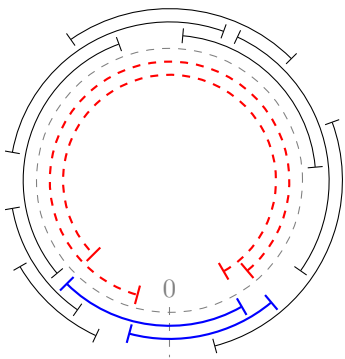

\begin{proof}[Proof of Theorem~\ref{thm:CAtractable}]
We consider each arc of $\ca A$ in angular coordinates as $[\alpha,\beta]$
clockwise, where $\alpha,\beta\in[0,2\pi)$.
By standard arguments (a ``small perturbation''), 
we can assume that no two arcs share the same endpoint,
and no arc starts or ends in (the angle)~$0$.
Let $\ca A_0\subseteq \ca A$ denote the subset of arcs containing~$0$.
Note that for every arc $[\alpha,\beta]\in\ca A_0$ we have $\alpha>\beta$,
and we subsequently define
$\ca A_1:=\big\{[\beta,\alpha]: [\alpha,\beta]\in\ca A_0\big\}$
as the set of their ``complementary'' arcs avoiding~$0$.
For $a\in\ca A_0$ we shortly denote by $\bar a\in\ca A_1$ its complementary arc.

Now, the set $\ca B:=(\ca A\sem\ca A_0)\cup \ca A_1$ is an ordinary
interval representation contained in the open line segment $(0,2\pi)$.
See Figure~\ref{fig:arc-complement}.
Since each of $\ca A\sem\ca A_0$ and $\ca A_1$ is $k$-fold proper by the
assumption on~$\ca A$, the representation $\ca B$ is $2k$-fold proper.
Note the following facts;
every two intervals in $\ca A_0$ intersect, and
an interval $a\in\ca A_0$ intersects $b\in\ca A\sem\ca A_0$
iff $b\not\subset\bar a$.

We now apply Lemma~\ref{INTtoposet} to the set~$\ca B$, constructing a
(labelled) poset $\ca P$ of width at most~$2k+1$.
We also add a new label $red$ to the elements of $\ca P$
which represent the arcs in $\ca A_1$.
The final step will give a definition of an \FO interpretation
$I=(\nu,\psi_1)$ such that~$I(\ca P)$ will be isomorphic to the intersection
graph $G$ of~$\ca A$.
Using the formulas $\psi,\vartheta$ from Lemma~\ref{INTtoposet}, the latter
is also quite easy.
As mentioned above, intersecting pairs of intervals from $\ca A$
can be described using intersection and containment of
the corresponding intervals of $\ca B$:
$$
\psi_1(x,y)\>\equiv\> \big(red(x)\wedge red(y)\big) \vee
	 \big(\neg red(x)\wedge\neg red(y)\wedge \psi(x,y)\big) \vee
	 \big(red(x)\wedge\neg red(y)\wedge\neg \vartheta(y,x)\big)
$$
It is routine to verify that, indeed, $G\simeq I(\ca P)$
(using the obvious bijection of $\ca A_0$ to $\ca A_1$). 

We then finish simply by Theorem~\ref{thm:posetFPT} and
Proposition~\ref{prop:reducetoposet}.
\end{proof}

\shortlong{}{%
One can speculate whether the parameter $k$ in Theorem~\ref{thm:CAtractable}
can be replaced by a number which is ``directly observable'' from the graph $G$, 
such as the maximum clique size.
However, the idea of taking the maximum clique size 
as such a parameter is not a brilliant idea since circular-arc graphs of bounded clique
size also have bounded tree-width, and so their \FO model checking becomes
easy by traditional means.
On the other hand, considering independent set size as an additional
parameter does not work either, as we will see in
Section~\ref{sec:hardness}.
}

\subsection{Circle graphs}\label{sec:CIR}

Another graph class closely related to interval graphs
are {\em circle graphs}, also known as {\em interval overlap graphs}.
These are intersection graphs of chords of a circle, and they can
equivalently be characterised as having an {\em overlap} interval representation 
$\ca C$ such that $a,b\in\ca C$ form an edge, if and only if
$a\cap b\not=\emptyset$ but neither $a\subseteq b$ nor $b\subseteq a$ hold
(see Figure~\ref{fig:open-circlerep}).
A circle representation of a circle graph can be efficiently
constructed~\cite{zbMATH04089581}.

Related {\em permutation graphs}
are defined as intersection graphs of line segments with the ends on
two parallel lines, and they form a complementation-closed subclass of circle graphs.
Note another easy characterization:
let $G$ be a graph and $G_1$ be obtained by adding one vertex adjacent to
all vertices of~$G$; then $G$ is a permutation graph if and only if $G_1$
is a circle graph.
We will see in Section~\ref{sec:hardness} that
the \EFO model checking problem is $W[1]$-hard for circle graphs,
and the \FO model checking
problem is $AW[*]$-complete already for {permutation graphs}.
However, there is also a positive result using a natural additional
parameterization.
\shortlong{The proof of it uses arguments similar to those of
Theorem~\ref{thm:CAtractable}.}{}

\begin{figure}[t]
$$
\begin{tikzpicture}[scale=\shortlong{0.8}{1.0}]
    \draw[dashed] (0,0) circle[radius=2cm];
    \draw (-90:2*0.95) --  (-90:2/0.95) node[below] {0};
    \draw (-130:2) -- (30:2);
    \node (b1) at (-130:2cm+0.7em) {$b_1$};
    \node (b2) at (30:2cm+0.7em) {$b_2$};
    \draw (150:2) -- (50:2);
    \node (d1) at (150:2cm+0.7em) {$d_1$};
    \node (d2) at (50:2cm+0.7em) {$d_2$};
    \draw (70:2) -- (-40:2);
    \node (f1) at (70:2cm+0.7em) {$f_1$};
    \node (f2) at (-40:2cm+0.7em) {$f_2$};
    \draw (-110:2) -- (170:2);
    \node (a1) at (-110:2cm+0.7em) {$a_1$};
    \node (a2) at (170:2cm+0.7em) {$a_2$};
    \draw (-170:2) -- (110:2);
    \node (c1) at (-170:2cm+0.7em) {$c_1$};
    \node (c2) at (110:2cm+0.7em) {$c_2$};
    \draw (130:2) -- (-20:2);
    \node (e1) at (130:2cm+0.7em) {$e_1$};
    \node (e2) at (-20:2cm+0.7em) {$e_2$};
    \draw (0:2) -- (-60:2);
    \node (g1) at (0:2cm+0.7em) {$g_1$};
    \node (g2) at (-60:2cm+0.7em) {$g_2$};
\end{tikzpicture}
\qquad
\begin{tikzpicture}[scale=\shortlong{0.4}{0.4}]
\draw[dashed] (0,0) node[below] {0} --
      (1,0) node[below] {$a_1$} --
      (2,0) node[below] {$b_1$} --
      (5,0) node[below] {$c_1$} --
      (6,0) node[below] {$a_2$} --
      (7,0) node[below] {$d_1$} --
      (8,0) node[below] {$e_1$} --
      (9,0) node[below] {$c_2$} --
      (11,0) node[below] {$f_1$} --
      (12,0) node[below] {$d_2$} --
      (13,0) node[below] {$b_2$} --
      (15,0) node[below] {$g_1$} --
      (16,0) node[below] {$e_2$} --
      (17,0) node[below] {$f_2$} --
      (18,0) node[below] {$g_2$} --
      (20,0) node[below] {$2\pi$};
      \draw (1,0) to [out=40,in=140] (6,0);
      \draw (2,0) to [out=70,in=180] (7.5,4.9) to [out=0,in=110](13,0);
      \draw (5,0) to [out=35,in=145] (9,0);
      \draw (7,0) to [out=40,in=140] (12,0);
      \draw (8,0) to [out=60,in=120] (16,0);
      \draw (11,0) to [out=45,in=135] (17,0);
      \draw (15,0) to [out=30,in=150] (18,0);
\end{tikzpicture}
$$
\caption{``Opening'' a circle representation
	(left; an intersecting system of chords of a circle)
	into an overlap representation
	(right; the depicted arcs to be flattened into intervals on the line).}
\label{fig:open-circlerep}
\end{figure}
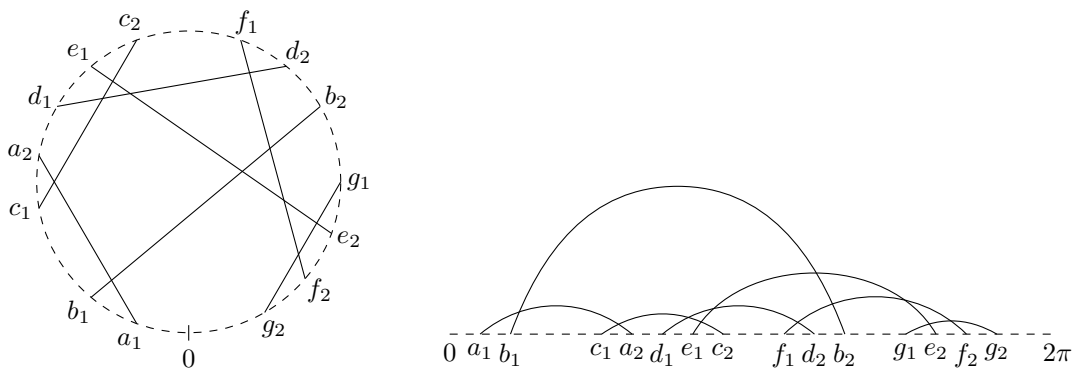

\begin{thm}\apxmark
\label{thm:CIRitractable}
The \FO model checking problem of circle graphs is FPT with respect 
to the formula and the maximum independent set size.
\end{thm}

\shortlong{}{Our proof is again closely based on Lemma~\ref{INTtoposet}, 
as in the previous section.
\begin{proof}
Let $G$ be an input circle graph.
We use, e.g., \cite{zbMATH04089581} to construct a set of chords
$\ca C$ such that $G$ is the intersection graph of~$\ca C$.
Again, by a small perturbation, we may assume that no two ends of chords
coincide.
Every chord $a\in\ca C$ can be specified as a pair $a=(\alpha,\beta)$
where $\alpha,\beta\in[0,2\pi)$ are the angular coordinates
of the endpoints of~$a$.
We define a set $\ca B:=\{[\alpha,\beta]:(\alpha,\beta)\in\ca C\}$
of intervals on $[0,2\pi)$,
which is an overlap representation of~$G$.

Let $k>0$ be such that the set $\ca B$ is $k$-fold proper.
Then $k$ is a lower bound on an independent set size in~$G$.
From Lemma~\ref{INTtoposet} applied to~$\ca B$, we get a poset
$\ca P$, and the formulas $\nu,\psi,\vartheta$ depending on $k$.
By the definition of an overlap representation, we can write
$$
\sigma(x,y)\>\equiv\> \psi(x,y)\wedge \neg\vartheta(x,y)
		\wedge \neg\vartheta(y,x)
$$
such that $I=(\nu,\sigma)$ is an \FO interpretation
satisfying $G\simeq I(\ca P)$.
Let $\ell$ be the maximum independent set size in $G$
(which we do not need to explicitly know).
Then the width of $\ca P$ is at most $k+1\leq\ell+1$, and so
we again finish simply by Theorem~\ref{thm:posetFPT} and
Proposition~\ref{prop:reducetoposet}.
\end{proof}
}

An interesting question is whether `independent set size' in
Theorem~\ref{thm:CIRitractable} can also be replaced with `clique size'.
We think the right answer is `yes', but we have not yet found the algorithm.
At least, the answer is positive for the subclass of permutation graphs:

\begin{cor}\apxmark
\label{cor:PERMtractable}
The \FO model checking problem of permutation graphs is FPT with respect 
to the formula size, and either the maximum clique or the maximum independent set size.
\end{cor}
\shortlong{}{\begin{proof}
Given a permutation graph $G$, we can efficiently construct its
representation~\cite{zbMATH03907788}.
Notice that reversing one line of this representation makes a representation
of the complement~$\overline G$.
Subsequently, an easy algorithm can compute, using the permutation
representations, the maximum independent sets of $G$ and of $\overline G$.
For the smaller one, we run the algorithm of Theorem~\ref{thm:CIRitractable}.
\end{proof}}

\begin{cor}\apxmark
The subgraph isomorphism (not induced) problem of permutation graphs 
is FPT with respect to the subgraph size.
\end{cor}
\shortlong{}{%
\begin{proof}
For a permutation graph $G$ and parameter $H$, we would like to decide
whether $H\subseteq G$.
If $G$ contains a $|V(H)|$-clique (which can be easily tested on permutation
graphs), then the answer is `yes'.
Otherwise, we answer by Corollary~\ref{cor:PERMtractable}.
\end{proof}
}

\subsection{Box and disk graphs}\label{sec:boxdisk}

{\em Box (intersection) graphs} are graphs having an intersection
representation by rectangles in the plane, such that each rectangle (box) 
has its sides parallel to the x- and y-axes.
The recognition problem of box graphs is NP-hard~\cite{zbMATH03815706}, and
so it is essential that the input of our algorithm would consist of a
box representation.
{\em Unit-box graphs} are those having a representation by unit boxes.

The \EFO model checking problem is $W[1]$-hard already for unit-box
graphs~\cite{DBLP:conf/esa/Marx05},
and we will furthermore show that 
it stays hard if we restrict the representation to a small area
in Proposition~\ref{pro:EFOhardness}.
Here we give the following slight extension of Theorem~\ref{thm:INTtractable}:

\begin{thm}\apxmark
\label{thm:Boxtractable}
Let $G$ be a box intersection graph given alongside
with its box representation $\ca B$ such that the following holds:
the projection of $\ca B$ to the x-axis is a
$k$-fold proper set of intervals, and the projection of $\ca B$ to the
y-axis consists of at most $k$ distinct intervals.
Then \FO model checking of $G$ is FPT with respect to the parameters 
$k$ and the formula size.
\end{thm}

\shortlong{}{%
\begin{figure}[t]
$$
\begin{tikzpicture}[scale=2.2]
\tikzstyle{every path}=[draw,color=red,
			fill={rgb,255:red,255; green,220; blue,220}];
\def\redbox{\draw (0.2,0.2) rectangle +(1,0.5) ;
	\draw (0.4,0.9) rectangle +(0.6,0.5) ;
	\draw (1.1,0.5) rectangle +(0.9,0.7) ;
	\draw (1.5,0.2) rectangle +(0.3,0.5) ;
	\draw (1.4,0.9) rectangle +(0.5,0.5) ;
}\redbox
\def\axes{
    \draw[->] (-0.2,0) -- (2.2,0) node[below] {$x$};
    \draw[->] (0,-0.2) -- (0,1.6) node[left] {$y$};
}
\def\interx{
    \draw[|-|] (-0.1,0.2) -- node[left,pos=0.2] {$t_1$} (-0.1,0.7);
    \draw[|-|] (-0.2,0.5) -- node[left] {$t_2$} (-0.2,1.2);
    \draw[|-|] (-0.1,0.9) -- node[left,pos=0.8] {$t_3$} (-0.1,1.4);
}
\tikzstyle{every path}=[draw,color=red,dashed,fill=none];
\redbox
\tikzstyle{every path}=[draw,color=black,fill=none];
\axes\interx
\end{tikzpicture}
\qquad
\begin{tikzpicture}
\tikzstyle{vecArrow} = [thick, decoration={markings,mark=at position
   1 with {\arrow[semithick]{open triangle 60}}},
      double distance=1.4pt, shorten >= 5.5pt,
         preaction = {decorate},
        postaction = {draw,line width=1.4pt, white,shorten >= 4.5pt}
        ];
\draw[vecArrow] (0,2.3) -- (0.8,2.3);
\node (a) at (0,0) {};
\end{tikzpicture}
\qquad
\begin{tikzpicture}[scale=\shortlong{1}{0.5}]
\draw (0,0) -- (0,9);
\draw (0,0) -- (-3,2) -- (0,4);
\draw (0,3) -- (-3,8) -- (0,9);
\draw (0,6) -- (1,6.5) -- (0,7);
\draw (0,5) -- (3,6.5) -- (0,8);
\draw (0,1) -- (3,1.5) -- (0,2);
\draw[dashed] (3,1.5) -- (3,6.5);
\draw (-3,2) -- (-3,8);
\foreach \x in {0,...,9} {
    \draw[fill=white] (0,1.0*\x)  circle[radius=5pt];}
\draw[fill=red] (-3,2)  circle[radius=5pt] node[left,black] {$L_1$};
\draw[fill=red] (-3,8)  circle[radius=5pt] node[left,black] {$L_2$};
\draw[fill=red] (1,6.5)  circle[radius=5pt] node[right,black] {$L_1$};
\draw[fill=red] (3,6.5)  circle[radius=5pt] node[right,black] {$L_3$};
\draw[fill=red] (3,1.5)  circle[radius=5pt] node[right,black] {$L_3$};
\end{tikzpicture}
$$
\caption{An illustration of
constructing a poset from the box representation with parameter $k=3$
(cf.~Theorem \ref{thm:Boxtractable});
the projection of the boxes to the x-axis is a $3$-fold proper interval representation,
and their projection to the y-axis consists of three intervals $t_1,t_2,t_3$.
The projected intervals on the x-axis give raise to a poset
of width $4$ on the right, where the highlighted points (red) represent the
boxes and the labels $L_1,L_2,L_3$ annotate their projected intervals on the
y-axis.}
\label{fig:construct-boxes}
\end{figure}
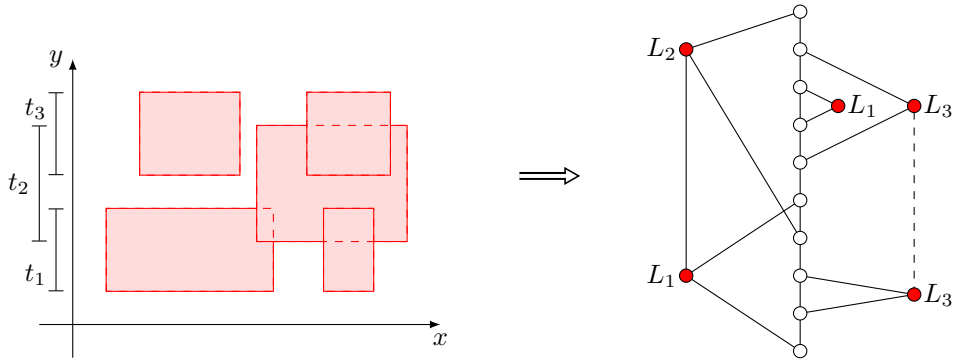

\begin{proof}
Let $\ca X$ be the set of intervals which are the projections of $\ca B$ to
the x-axis.
Again, we can, by a small perturbation in the x-direction, assume that no
two intervals from $\ca X$ share a common end.
Then we apply Lemma~\ref{INTtoposet} to~$\ca X$, and get a poset
$\ca P$ of width $\leq k+1$ and the formulas $\nu,\psi$ depending on $k$.
See Figure~\ref{fig:construct-boxes}.
In addition to the previous, we number the distinct intervals to which 
$\ca B$ projects onto the y-axis, as $t_1,\dots,t_\ell$ where 
$\ell\leq k$.
We give label $L_i$ to each box of $\ca B$ which projects onto~$t_i$.
Then we define
$$
\sigma(x,y)\>\equiv\>  \psi(x,y)\wedge\left[
	\bigvee\nolimits_{1\leq i,j\leq \ell:~ t_i\cap t_j\not=\emptyset}
		\big( L_i(x)\wedge L_j(y) \big)
\right],
$$
meaning that the projections of the boxes $x$ and $y$ intersect on the x-axis,
and moreover their projections onto the y-axis are also intersecting.
Hence, for the \FO interpretation $I=(\nu,\sigma)$, we have got
$I(\ca P)\simeq G$.
We again finish by Theorem~\ref{thm:posetFPT} and
Proposition~\ref{prop:reducetoposet}.
\end{proof}
Note that the idea of handling projections to the x-axis as an interval
graph by Lemma~\ref{INTtoposet} {\em cannot} be simultaneously applied to
the y-axis.
The reason is that the two separate posets (for x- and y-axes), 
sharing the boxes as their common elements, would not together form a poset.
Another strong reason is given in Corollary~\ref{cor:allhardness}c).
\smallskip
}

Furthermore, {\em disk graphs} 
are those having an intersection representation by disks in the plane.
Their recognition problem is NP-hard already with unit disks~\cite{kirk_unit_98},
and the \EFO model checking problem is $W[1]$-hard again for unit-disk
graphs by~\cite{DBLP:conf/esa/Marx05}.
Similarly to Theorem~\ref{thm:Boxtractable}, we have identified 
a tractable case of \FO model checking of unit-disk graph, based on
restricting the y-coordinates of the disks.
\shortlong{%
Due to space restrictions, we leave this case only for the full paper.
}{%

\begin{thm}
\label{thm:Disktractable}
Let $G$ be a unit-disk intersection graph given alongside
with its unit-disk representation $\ca B$ such that
the disks use only $k$ distinct y-coordinates.
Then \FO model checking of $G$ is FPT with respect to the parameters 
$k$ and the formula size.
\end{thm}

\begin{proof}
For start, note that we cannot use here the same easy approach as in the
proof of Theorem~\ref{thm:Boxtractable}, since
one cannot simply tell whether two disks intersect from the intersection of
their projections onto the axes.
Instead, we will use the following observation:
if two unit disks, with the y-coordinates $y_1,y_2$ of their centers,
intersect each other, then they do so in some point at the y-coordinate
$\frac12(y_1+y_2)$.

By the assumption,
let $\ca B=\ca B_1\cup\dots\cup\ca B_k$ such that all disks in $\ca B_i$
have their centers at the y-coordinate $y_i$, for $i\in\{1,\dots,k\}$.
For each $i,j\in\{1,\dots,k\}$ (not necessarily distinct),
we define a set of intervals $\ca X_{i,j}$ which are the intersections of
the disks from $\ca B_i\cup\ca B_j$ with the horizontal line given by
$y=\frac12(y_i+y_j)$.
Note that $\ca X_{i,j}$ is proper since all our disks are of the same size,
and that two disks from $\ca B$ intersect if and only if their corresponding
intervals in some $\ca X_{i,j}$ intersect.
Again, by a standard argument of small enlargement and perturbation
of the disks, 
we may assume that all the interval ends in $\ca X_{i,j}$ are distinct.

Then we apply Lemma~\ref{INTtoposet} to each~$\ca X_{i,j}$, and get posets
$\ca P_{i,j}=(P_{i,j},\leq^{i,j})$ of width $2$.
By the natural correspondence between the disks of $\ca B_i\cup\ca B_j$
and the intervals of $\ca X_{i,j}$, we may actually assume
that $\ca B_i\cup\ca B_j\subseteq P_{i,j}$
and $\ca B_i\cup\ca B_j$ is linearly ordered in $\ca P_{i,j}$
according to the x-coordinates of the disks.
We linearly order~$\ca B$ by the x-coordinates of the disks
and, with respect to this ordering, we make the union
$\ca P:=\bigcup_{1\leq i,j\leq k}\ca P_{i,j}$
and apply transitive closure.
Then $\ca P$ is a poset of width~$k^2+1$.
We also give, for each $i,j$, a label $B_i$ to the elements of $\ca B_i$
in $\ca P$ and a label $D_{i,j}$ to the elements of 
$P_{i,j}\sem(\ca B_i\cup\ca B_j)$ in~$\ca P$.

It remains to define an \FO interpretation $I=(\nu,\psi)$
such that $I(\ca P)\simeq G$.
For that we straightforwardly adapt the formulas from Lemma~\ref{INTtoposet}:
\begin{align*}
\nu(x) \equiv\>& \bigvee\nolimits_{1\leq i\leq k}B_{i}(x)
\\
\psi(x,y) \equiv\>& \bigvee\nolimits_{1\leq i,j\leq k} \left[\,
	B_{i}(x)\wedge B_j(y)\wedge \vbox to 2.5ex{\vfill}\right.
\\	&\qquad\left.\vbox to 2.5ex{\vfill}
	 \forall z\, \left[ D_{i,j}(z) \to \left(
		(\neg\, x\leq^\PPP\!z \vee \neg\, z\leq^\PPP\!y) \wedge
		(\neg\, y\leq^\PPP\!z \vee \neg\, z\leq^\PPP\!x)
	\right)\right]\right]
\end{align*}
By the assigned labelling ($B_i$ and $D_{i,j}$),
$\ca P\models\psi(u,v)$ if and only if
there are $i,j$ such that $u\in B_i$, $v\in B_j$
and the corresponding intervals in $\ca X_{i,j}$ intersect.
That is, iff $uv\in E(G)$.
\end{proof}
}

\begin{figure}
 \centering
 \includegraphics[width=0.6\textwidth]{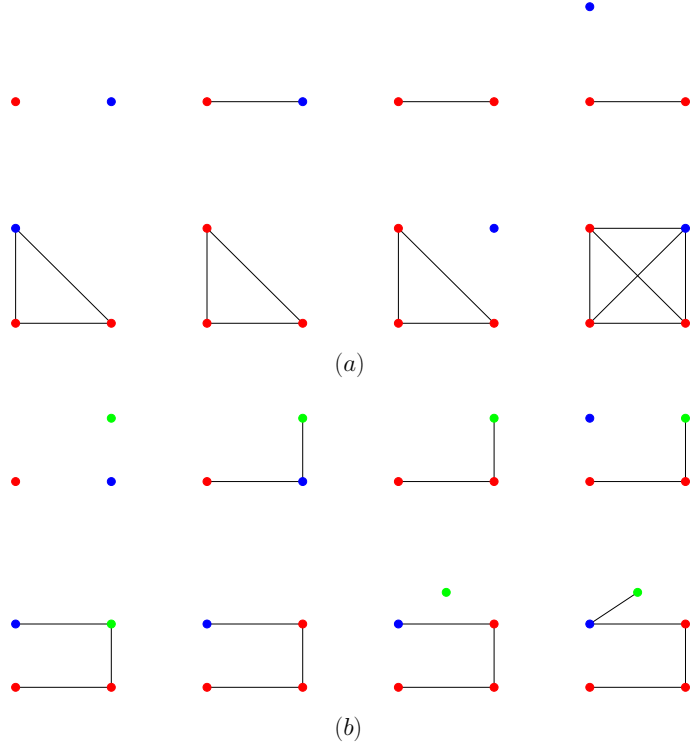}
\medskip
 \caption{The operations for obtaining the clique-width of a graph, illustrated for (a) $K_4$, which has clique-width $2$, and (b) $P_5$,
 which has clique-width $3$.
  }
  \label{figcw}
\end{figure}

\subsection{Unbounded local clique-width}\label{sec:unboundedcw}

A $k$-labelled graph is a graph whose vertices are assigned integers (called
labels) from $1$ to $k$ (each vertex has precisely one label).  The {\em
  clique-width} of a graph $G$ equals the minimum $k$ such that $G$ can be
obtained using the following four operations: creating a vertex labelled 
$1$, relabeling all vertices with label $i$ to label $j$, adding all edges between
the vertices with label $i$ and the vertices with label $j$, and
taking a disjoint union of graphs obtained using these operations (see Figure \ref{figcw}).

We say that a graph class $\scri C$ is of {\em bounded local clique-width}
if there exists a function $g$ such that the following holds:
for every graph $G\in\scri C$, integer $d$ and every vertex $x$ of $G$,
the clique-width of the subgraph of $G$ induced on the vertices at distance
$\leq d$ from $x$ in at most~$g(d)$.

As noted above, the algorithm of~\cite{cmr00} for MSO on graphs of bounded 
clique-width implies fixed-parameter tractability of \FO model checking 
on graphs of bounded local clique-width via Gaifman's locality.
Though, the following opposite result was shown in \cite{GHKOST15}
(note that the considered graph class has bounded diameter, and so claiming
unbounded clique-width is enough):

\begin{prop}[{\cite[Proposition~5.2]{GHKOST15}}]
For any irrational $q>0$ there is $\ell$ such that the subclass of
interval graphs represented by intervals of lengths $1$ and $q$ on a line
segment of length $\ell$ has unbounded clique-width.
\end{prop}

This immediately implies unbounded local clique-width for classes
of $2$-fold proper circular-arc graphs and also for similar subclasses of
box and disk graphs, which justifies relevance of our new algorithms for 
\FO model checking.
Moreover, by an adaptation of the core idea of
\cite[Proposition~5.2]{GHKOST15} we can prove a much stronger negative result.
We start with a claim capturing the essence of the construction in
\cite[Proposition~5.2]{GHKOST15}.

Consider disjoint $m$-element vertex sets $X$ and $Y$ in a graph.
We say that $X$ is {\em gradually connected} with~$Y$
if there exist orderings $X=\{x_1,x_2,\dots,x_m\}$
and $Y=\{y_1,y_2,\dots,y_m\}$ such that, for any $i<j$,
$x_jy_i$ is an edge while $x_iy_j$ is not an edge (we do not care about
edges of the form $x_iy_i$).
Recall that a {\em transversal} of a set system $\{X_1,X_2,\dots,X_r\}$
is a set $Z=\{z_1,\dots,z_r\}$ of $r$ distinct elements
such that $z_i\in X_i$ for~$i=1,\dots,r$.

\begin{lem}\label{lem:unboundedcw}
Let $k$ be an integer.
Let $G$ be a graph and $V_1,V_2,\dots,V_r$ be a partition of the vertex set
of $G$ such that $|V_1|=|V_2|=\dots=|V_r|=m$, and $m>6kr$.
Assume that $V_i$ is gradually connected with $V_{i+1}$ for
$i=1,2,\dots,r-1$.
Furthermore, assume that there exists a set $I\subseteq\{1,\dots,r\}$,
$|I|=2k$, such that the following holds:
for any sets $X,Y$ such that $X$ is a transversal of the set system
$\{V_i:i\in I\}$ and $Y$ is a transversal of $\{V_{i+1}:i\in I\}$,
the set $X$ is gradually connected to~$Y$.
Then the clique-width of $G$ is at least~$k$.
\end{lem}
\begin{proof}
Let $G$ be an assumed graph on $n=rm$ vertices, but the clique-width of $G$ is at most $k-1$.
In the construction of $G$ using $k-1$ labels from the definition of clique-width,
a $(k-1)$-labelled subgraph $G_1$ of $G$ with 
$\frac13n\leq |V(G_1)| \leq\frac23n$ must have appeared.
We will now get a contradiction by showing a set of $k$ vertices
of $G_1$ which have pairwise different neighbourhoods in $G- V(G_1)$.

Suppose that there exists $i$ such that $|V_{i+1}\cap V(G_1)|-|V_{i}\cap V(G_1)|\geq2k$.
Then there are sets $X\subseteq V_{i}\setminus V(G_1)$ and
$Y\subseteq V_{i+1}\cap V(G_1)$, where $|X|=|Y|=2k$, such that
$X$ is gradually connected to $Y$ with respect to orderings
$X=\{x_1,\dots,x_{2k}\}$ and $Y=\{y_1,\dots,y_{2k}\}$.
Then the vertices $y_1,y_3,\dots,y_{2k-1}$ of $G_1$ have pairwise different
neighbourhoods in $G- V(G_1)$, as witnessed by $x_2,x_4,\dots,x_{2k}$.
The same applies if $|V_{i+1}\cap V(G_1)|-|V_{i}\cap V(G_1)|\leq-2k$.

The next step is to show that, for $i=1,\dots,r$,
it holds $\emptyset\not=V_{i}\cap V(G_1)\not=V(G_1)$.
Indeed, up to symmetry, let $V_{i}\cap V(G_1)=\emptyset$ for some~$i$,
which implies $|V_{j}\cap V(G_1)|<2kr<\frac13m$ for all $j\in\{1,\dots,r\}$ 
by the previous paragraph, and the latter contradicts our assumption 
$|V(G_1)|\geq\frac13n$.
For the assumed index set $I\subseteq\{1,\dots,r\}$,
we can hence choose sets $X$ a transversal of $\{V_i:i\in I\}$ 
and $Y$ a transversal of $\{V_{i+1}:i\in I\}$, such that
$X\cap V(G_1)=\emptyset$ and $Y\subseteq V(G_1)$.
Moreover, $|X|=|Y|=2k$ and $X$ is gradually connected to $Y$ 
by the assumption of the lemma. We thus again get a contradiction as above.
\end{proof}

\begin{prop}\label{prop:unboundedcw}
The following graph classes contain subclasses of bounded diameter
and unbounded clique-width:
\begin{itemize}\parskip0pt
\item unit circular-arc graphs of independence number $2$,
\item circle graphs of independence number $2$,
\item unit box and disk graphs with a representation contained within a
square of bounded size.
\end{itemize}
\end{prop}

\begin{figure}[t]
$$
\begin{tikzpicture}[scale=1.2]
\tikzstyle{every path}=[draw,color=gray];
\draw[dashed] (0,0) circle (15mm);
\tikzstyle{every path}=[draw,color=black, |-|];
\centerarc[] (0,0) (-90:33:1.55) ;
\centerarc[] (0,0) (-88:35:1.58) ;
\centerarc[] (0,0) (-86:37:1.61) ;
\centerarc[] (0,0) (-81:42:1.82) ;
\centerarc[] (0,0) (-79:44:1.85) ;
\centerarc[] (0,0) (-77:46:1.88) ;
\centerarc[] (0,0) (-72:51:2.09) ;
\centerarc[] (0,0) (-70:53:2.12) ;
\centerarc[] (0,0) (-68:55:2.15) ;
\centerarc[dotted,thick,-] (0,0) (-60:60:2.35) ;
\tikzstyle{every path}=[draw,color=red, |-|];
\centerarc[] (0,0) (33:156:1.64) ;
\centerarc[] (0,0) (35:158:1.67) ;
\centerarc[] (0,0) (37:160:1.70) ;
\centerarc[] (0,0) (42:165:1.91) ;
\centerarc[] (0,0) (44:167:1.94) ;
\centerarc[] (0,0) (46:169:1.97) ;
\centerarc[dotted,thick,-] (0,0) (65:171:2.15) ;
\tikzstyle{every path}=[draw,color=green, |-|];
\centerarc[] (0,0) (156:279:1.73) ;
\centerarc[] (0,0) (158:281:1.76) ;
\centerarc[] (0,0) (160:283:1.79) ;
\centerarc[] (0,0) (165:288:2.00) ;
\centerarc[] (0,0) (167:290:2.03) ;
\centerarc[] (0,0) (169:292:2.06) ;
\centerarc[dotted,thick,-] (0,0) (180:293:2.25) ;
\end{tikzpicture}
\qquad\qquad\qquad
\begin{tikzpicture}[scale=1.16]
\def\rectIII{rectangle +(2,2) ++(0.04,0.04) rectangle +(2,2) 
			++(0.04,0.04) rectangle +(2,2)}
\tikzstyle{every path}=[draw,color=green,
			fill={rgb,255:red,220; green,255; blue,220}];
	\draw (1,2.2) \rectIII ;
	\draw (1.2,2.4) \rectIII ;
\tikzstyle{every path}=[draw,color=red,
			fill={rgb,255:red,255; green,220; blue,220}];
	\draw (2,0.2) \rectIII ;
	\draw (2.2,0.4) \rectIII ;
	\draw (2.4,0.6) \rectIII ;
\tikzstyle{every path}=[draw,color=black,
			fill={rgb,255:red,220; green,220; blue,220}];
	\draw (0,0) \rectIII ;
	\draw (0.2,0.2) \rectIII ;
	\draw (0.4,0.4) \rectIII ;
\tikzstyle{every path}=[draw,color=red, dashed, fill=none];
	\draw (2,0.2) \rectIII ;
	\draw (2.2,0.4) \rectIII ;
	\draw (2.4,0.6) \rectIII ;
\tikzstyle{every path}=[draw,color=black, dashed,
		dash pattern=on 1pt off 3pt on 2pt off 0pt, fill=none];
	\draw (0,0) \rectIII ;
	\draw (0.2,0.2) \rectIII ;
\tikzstyle{every path}=[draw,color=green, dashed,
		dash pattern=on 0pt off 3pt on 3pt off 0pt, fill=none];
	\draw (1,2.2) \rectIII ;
	\draw (1.2,2.4) \rectIII ;
\tikzstyle{every path}=[draw,color=black, very thick,dotted]
\draw (1.8,1.8) -- (2.8,1.8) ;
\end{tikzpicture}
$$
\caption{An illustration of the constructions used in the proof of
Proposition~\ref{prop:unboundedcw}. Left: unit circular-arc graphs
for $m=3$ (this is not a valid value according to the proof, 
but proper $m=36k+1$ would not produce a comprehensible picture).
Right: unit box graphs for~$m=3$.
The arcs/boxes in black colour represent the sets $V_1,V_4,V_7,\dots$.}
\label{fig:unboundedcw}
\end{figure}
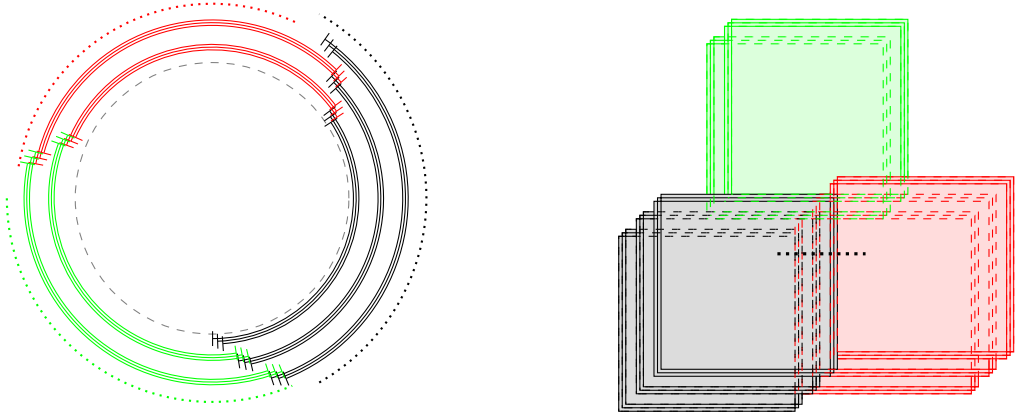

\begin{proof}
Our overall aim is to construct special intersection representations 
of graphs within the claimed classes, which have bounded diameters
and whose vertex sets can be partitioned into sets $V_1,V_2,\dots,V_r$ 
of properties assumed in Lemma~\ref{lem:unboundedcw}.
See also Figure~\ref{fig:unboundedcw}.

First, consider the circle of radius $1$ and ``unit'' arcs of fixed length 
$a=(2\pi+\delta)/3$ on this circle, for a sufficiently small~$\delta>0$.
Since $3a$ is more than the circumference of the circle,
there cannot be three disjoint arcs and the diameter of any
such intersection graph is at most~$3$.
Choose $r=6k$ and~$m=36k+1$,
and let $\varepsilon$ be such that $0<\varepsilon<\delta/m$.
Let $V_1$ consist of $m$  arcs of length $a$ starting at angles
$0,\varepsilon,\dots,(m-1)\varepsilon$, and let $V_i$ for $i=2,\dots,r$
be a copy of $V_1$ shifted by the angle $(i-1)a$ counterclockwise.
Clearly, $V_{i-1}$ is gradually connected with $V_i$.
Moreover, $V_4$ is in fact a copy of $V_1$ shifted by the angle $\delta$
counterclockwise, and analogously with $V_7,V_{10}$, etc.
Assuming $(r+1)\delta<3$, this means that the whole set
$V_1\cup\ V_4\cup\dots\cup V_{r-2}$ is gradually connected with
$V_2\cup\ V_5\cup\dots\cup V_{r-1}$, and this implies the conditions
of Lemma~\ref{lem:unboundedcw} with $I=\{1,4,\dots,r-2\}$.
Since $k$ can be chosen arbitrary, our graphs have unbounded clique-width.

The same construction can be used in the second case as well;
we simply replace each arc with a chord between the same ends,
and this circle representation would represent an isomorphic graph to the
previous case.

We have a similar construction also in the last two cases.
We again choose $r=6k$, $m=36k+1$ and $0<\varepsilon<\delta/m$.
Let the set $V_1$ consist of unit squares with the lower left corners at
coordinates $(i\varepsilon,i\varepsilon)$ where $i=0,1,\dots,m-1$.
Let $V_2$ be a copy of $V_1$ translated by the vector $(1,\delta)$
and $V_3$ be a copy of $V_1$ translated by $(\frac12,1+\delta)$.
For $j=3,6,\dots,r-3$, let the triple $V_{j+1},V_{j+2},V_{j+3}$
be a copy of $V_1,V_2,V_3$ translated by $(\delta,\delta)$.
Again, in the intersection graph, 
$V_{i-1}$ is gradually connected with $V_i$ for $i=1,2,\dots,r$.
Assuming $(r+1)\delta<\frac32$, we similarly fulfill the conditions
of Lemma~\ref{lem:unboundedcw} with $I=\{1,4,\dots,r-2\}$,
thus proving our claim of unbounded clique-width.
One can also apply a similar construction in the case of unit disk graphs.
\end{proof}

\section{Hardness for Intersection Classes}\label{sec:hardness}

Our aim is to provide a generic reduction for proving hardness of \FO model
checking (even without labels on vertices) using only a simple property 
which is easy to establish for many geometric intersection graph classes.
We will then use it to derive hardness of \FO for quite restricted forms of
intersection representations studied in our paper
(Corollary~\ref{cor:allhardness}).

We say that a graph $G$ {\em represents consecutive neighbourhoods} of
order~$\ell$, if there exists a sequence 
$S=(v_1,v_2,\dots,v_\ell)\subseteq V(G)$ of distinct vertices of $G$ and
a set $R\subseteq V(G)$, $R\cap S=\emptyset$,
such that for each pair $i,j$, $1\leq i<j\leq \ell$,
there is a vertex $w\in R$ whose neighbours in $S$ are precisely the
vertices $v_i,v_{i+1}\dots,v_j$.
(Possible edges other than those between $R$ and $S$ do not matter.)
A graph class $\ca G$ has the {\em consecutive neighbourhood
representation} property if, for every integer $\ell>0$,
there exists an efficiently computable graph $G\in\ca G$ 
such that $G$ or its complement $\overline G$ represents 
consecutive neighbourhoods of order~$\ell$.

Note that our notion of `representing consecutive neighbourhoods'
is related to the concepts of ``$n$-order property'' and
``stability'' from model theory (mentioned in Section~\ref{sec:intro}).
This is not a random coincidence, as it is known
\cite{DBLP:journals/ejc/AdlerA14}
that on monotone graph classes stability coincides with nowhere dense
(which is the most general characterization allowing for FPT 
\FO model checking on monotone classes).
In our approach, we stress easy applicability of this notion to 
a wide range of geometric intersection graphs and,
to certain extent, to \EFO model checking.

The main result is as follows.
A {\em duplication} of a vertex $v$ in $G$ is the operation of~adding a
{\em true twin} $v'$ to $v$, i.e., 
new $v'$ adjacent to $v$ and precisely to the neighbours of $v$ in~$G$.

\begin{thm}\apxmark\label{thm:allhardness}
Let $\ca G$ be a class of unlabelled graphs having the consecutive neighbourhood
representation property, and $\ca G$ be closed on induced subgraphs and
duplication of vertices.
Then the \FO model checking of $\ca G$ is $AW[*]$-complete
with respect to the formula size.
\end{thm}

\shortlong{}{%
\begin{proof}
Our strategy is to prove that graphs in $\ca G$ 
can be used to represent any finite simple graph~$H$ ``via \FO'' --
using an \FO interpretation introduced in Section~\ref{sec:prelim}.
To this end, we give a pair of \FO formulas $I=(\nu,\psi)$ and
for any graph $H$, we efficiently construct graphs $G_H\in\ca G$ and 
$H'\simeq H$ such that $I(G_H)=H'$.
Precisely, the last expression means $V(H')=\{v: G_H\models\nu(v)\}$ and
$E(H')=\{uv: u,v\in V(H'),\, G_H\models\psi(u,v)\vee\psi(v,u)\}$.
Assuming this ($I(G_H)=H'\simeq H$) for a moment, 
we show how it implies the statement of the theorem.

Consider an \FO model checking instance on $\ca G$, 
parameterized by an \FO formula $\phi$.
We assume input $H$, and $I$ and $G_H\in\ca G$ as above,
and define an \FO formula $\phi^I$ recursively (cf.~Section~\ref{sec:prelim}):
every occurrence of ${edge}(x,y)$ is replaced by $\psi(x,y)\vee\psi(y,x)$,
every $\exists x\,\sigma$ is replaced by $\exists x\,(\nu(x)\wedge\sigma)$
and $\forall x\,\sigma$ by $\forall x\,(\nu(x)\to\sigma)$.
Clearly, $G_H\models\phi^I \Longleftrightarrow H\models\phi$.
The latter problem $H\models\phi$ is $AW[*]$-complete 
with respect to $|\phi|$ by Theorem~\ref{thm:FOgenhard}.
Since $|\phi^I|$ is bounded in $|\phi|$, we have got a parameterized
reduction implying that the \FO model checking problem of graphs from 
$\ca G$ is $AW[*]$-complete, too.
\medskip

Now we return to the initial task of defining the \FO
interpretation $I=(\nu,\psi)$ and constructing $G_H\in\ca G$ for given~$H$.
Let $V(H)=\{1,2,\dots,n\}$.
By the assumption, we can efficiently compute a graph $G_n\in\ca G$
that represents consecutive neighbourhoods of order~$n+2$,
as witnessed by a sequence $S=(v_0,v_1,\dots,v_n,v_{n+1})\subseteq V(G_n)$
and a set $R\subseteq V(G_n)$.
If it happened that, actually, the complement $\overline G_n$ represented
consecutive neighbourhoods, then we would simply switch to 
$\neg edge(x,y)$ in the formulas below.

For $0\leq i<j\leq n$, let $r_{i,j}\in R$ denote a vertex whose neighbours 
in $S$ are precisely $v_i,v_{i+1}\dots,v_j$.
Let $P:=\{r_{0,1},r_{1,2},\dots,r_{n,n+1}\}$
and $Q:=\{r_{i,j}: ij\in E(H)\}$ (it may happen that
$P\cap Q\not=\emptyset$, but $S\cap(P\cup Q)=\emptyset$).
We construct $G_H$ as the subgraph of $G_n$
induced on the vertex set $S\cup P\cup Q$.
By the assumption that $\ca G$ is closed on induced subgraphs,
we have got $G_H\in\ca G$.
Furthermore, we give labels `$blue$' to every vertex of $S$,
`$green$' to every vertex of $P$ and `$red$' to every vertex of $Q$
(those in $P\cap Q$ get both `$green$' and `$red$').

Using the labels,
construction of the desired \FO interpretation is now easy;
\begin{align}\label{eq:E-h-interpret}
\nu(x)&\equiv\> blue(x)\wedge
	\exists s,s' \big( s\not=s'\wedge
	green(s)\wedge green(s')\wedge edge(x,s)\wedge edge(x,s') \big)
,\\\nonumber
\psi(x,y)&\equiv\>  blue(x)\wedge blue(y)\wedge x\not=y
	\wedge \exists z \big[
		red(z)\wedge extreme(x,z)\wedge extreme(y,z)
	\big]
,\end{align}
where $\nu(x)$ is true precisely for $v_1,\dots,v_n$ of $S$, and
$extreme(x,z)$ in $\psi$ means that $x$ is one of the ``extreme'' neighbours
of $z$ within the sequence~$S$.
The point is that we can express the latter in \FO with help of the `$green$'
vertices which define the (symmetric) successor relation of~$S$ within the graph~$G_H$.
It is
\begin{align*}
extreme(x,z)\equiv~& edge(x,z)
\\&	\wedge \exists s,x' \big[
		green(s)\wedge blue(x')\wedge edge(x,s)
		\wedge edge(s,x')\wedge\neg edge(x',z)
	\big]
,\end{align*}
where the second line states that $x$ is connected to blue $x'$ via a green
vertex, such that $x'$ is not a neighbour of~$z$.
Altogether, for $I=(\nu,\psi)$ we easily verify $I(G_H)\simeq H$
(where the isomorphism maps each blue vertex $v_i$ to $i\in V(H)$).

The last step shows how we can get rid of the labels.
For that we use duplication of vertices (which preserves membership in 
$\ca G$ by the assumption).
For start, notice that no two vertices of $G_H$ can be twins by our construction.
Then every vertex in $P\sem Q$ is duplicated once, every vertex
in $P\cap Q$ is duplicated twice and every in $Q\sem P$ is duplicated three
times, forming the new graph~$G_H'\in\ca G$.

Regarding the formulas of $I$, we apply a corresponding transformation.
Start with a formula 
$twin(x,y) \equiv edge(x,y)\wedge \forall z [(z\neq x \wedge z\neq y)
 \rightarrow (edge(x,z)\leftrightarrow edge(y,z))]$
asserting that $x,y$ are true twins.
We can routinely write down formulas $dupl_d(x)$ asserting that
the vertex $x$ is a part of a class of $\geq d$ true twins, e.g.,
$dupl_2(x)\equiv \exists z (z\not=x\wedge twin(x,z))$
and $dupl_3(x)\equiv \exists z,z' (x\not=z\not=z'\not=x
 \wedge twin(x,z)\wedge twin(x,z'))$.
Then we transform $I=(\nu,\psi)$ into  $I'=(\nu',\psi')$ as follows
\begin{itemize}
\item $blue(x)$ is replaced with $\neg dupl_2(x)$,
\item $green(x)$ is replaced with $dupl_2(x)\wedge\neg dupl_4(x)$,
\item $red(x)$ is replaced with $dupl_3(x)$, and
\item $x=y$ is replaced with $twin(x,y)$.
\end{itemize}

One can routinely verify that again $I'(G_H')\simeq H$.
Moreover, $G_H'\in\ca G$ has been constructed in polynomial time from $H$,
and $G_H'$ carries no labels.
\end{proof}
}

\begin{figure}[t]
$$
\begin{tikzpicture}[scale=\shortlong{2.2}{2.8}]
\draw[dotted] (0,0) -- (1.5,0);
\draw[dotted] (0,1) -- (1.5,1);
\tikzstyle{every path}=[draw,color=black];
\foreach \x in {1,...,14} {
	\draw (0.1*\x,0) -- (0.1*\x,1) ;
}
\tikzstyle{every path}=[draw,color=red];
\draw (0.09,0) -- (0.31,1) ;
\draw (0.19,0) -- (0.51,1) ;
\draw (0.05,0) -- (1.45,1) ;
\draw (0.38,0) -- (0.52,1) ;
\draw (0.58,0) -- (0.92,1) ;
\draw (0.39,0) -- (1.22,1) ;
\draw (0.89,0) -- (1.21,1) ;
\draw (0.59,0) -- (1.31,1) ;
\draw (1.19,0) -- (1.41,1) ;
\end{tikzpicture}
\qquad\qquad\qquad
\begin{tikzpicture}[scale=\shortlong{1.1}{1.4}]
\tikzstyle{every path}=[draw,color=blue,
			 fill={rgb,255:red,220; green,220; blue,255}];
\def\bluebox{\foreach \x in {1,...,10} {
	\draw (0.1*\x,1-0.1*\x) rectangle (1+0.1*\x,2-0.1*\x) ;
}}\bluebox
\tikzstyle{every path}=[draw,color=red,
			fill={rgb,255:red,255; green,220; blue,220}];
\def\redbox{\draw (1.05,1.75) rectangle +(1,1) ;
	\draw (1.35,1.15) rectangle +(1,1) ;
	\draw (1.55,1.25) rectangle +(1,1) ;
	\draw (1.85,0.95) rectangle +(1,1) ;
}\redbox
\tikzstyle{every path}=[draw,color=black,fill=none];
\bluebox
\tikzstyle{every path}=[draw,color=red,fill=none];
\redbox
\end{tikzpicture}
$$
\caption{Constructing witnesses of the consecutive neighbourhood
representation property -- as permutation graphs (left)
and as unit-box graphs (right); cf.~Corollary~\ref{cor:allhardness}.}
\label{fig:construct-consecutiv}
\end{figure}
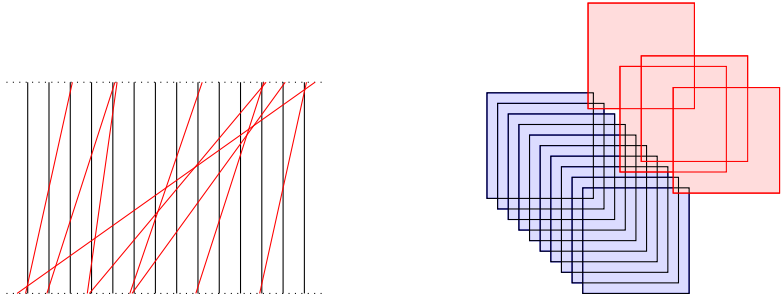

Graphs witnessing the consecutive neighbourhood representation property
can be easily constructed within our intersection classes,
even with strong further restrictions.
See some illustrating examples in Figure~\ref{fig:construct-consecutiv}.
So, we obtain the following hardness results:

\begin{cor}\apxmark
\label{cor:allhardness}
The \FO model checking problem is $AW[*]$-complete with respect to the formula size,
for each of the following geometric graph classes (all unlabelled):
\\a) circular-arc graphs with a representation consisting or arcs of lengths
from $[\pi-\varepsilon,\pi+\varepsilon]$ on the circle of diameter $1$,
for any fixed~$\varepsilon>0$,
\\b) connected permutation graphs,
\\c) unit-box graphs with a representation contained within a square of side
length $2+\varepsilon$, for any fixed~$\varepsilon>0$,
\\d) unit-disk graphs (that is of diameter~$1$) with a representation
contained within a rectangle of sides $1+\varepsilon$ and~$2$, 
for any fixed~$\varepsilon>0$.
\end{cor}

\shortlong{}{%
\begin{proof}
Each of the considered graph classes is routinely closed under induced
subgraphs and duplication.
Hence it is enough to construct, in each of the classes,
appropriate witnesses of the consecutive neighbourhood
representation property.

a) For an integer $n$ and $\delta:=\varepsilon/2n$,
we consider the sets of arcs 
$\ca S=\{[i\delta,\pi+i\delta]: i=1,2,\dots,n\}$
and $\ca N=\{[\pi+j\delta+\frac\delta2,\, i\delta-\frac\delta2]:
	1\leq i<j\leq n\}$.
The complement of the circular-arc intersection graph of $\ca S\cup\ca N$
represents consecutive neighbourhoods of order~$n$.

b) Let $x,y$ be two parallel lines.
We represent the line segments of a permutation representation
on $x,y$ by pairs $\langle x_i,y_i\rangle$ where $x_i,y_i$ are the
coordinates of the two ends on the lines $x,y$ respectively.
Our witness of order $n$ simply consists of the sets
$\ca S=\{\langle i,i\rangle: i=1,2,\dots,n\}$
and $\ca N=\{\langle i-\frac12,j+\frac12\rangle:
	1\leq i<j\leq n\}$,
as roughly depicted in Figure~\ref{fig:construct-consecutiv}.
\smallskip

c) As illustrated in Figure~\ref{fig:construct-consecutiv},
we specify $\ca S$ as the set of unit boxes $B_i$ with their lower corners
at coordinates $(i\delta,(n-i)\delta)$ where $i=1,2,\dots,n$
and $\delta:=\varepsilon/n$.
For any $1\leq i<j\leq n$, we introduce a unit box
with the lower left corner at 
$(1+i\delta-\frac\delta2, 1+(n-j)\delta-\frac\delta2)$,
which intersects exactly $B_i,B_{i+1},\dots,B_j$.
Let $\ca N$ denote the set of the latter boxes;
then the intersection graph of $\ca S\cup\ca N$
represents consecutive neighbourhoods of order~$n$.

d) We take the set $\ca S$ of unit disks $D_i$ (of diameter~$1$) 
with their centers at coordinates $(i\delta,0)$ where $i=1,2,\dots,n$     
and $\delta:=\varepsilon/n$.
Then, let $\ca N$ consists of the unit disks $D'_{i,j}$,
for all $1\leq i<j\leq n$, with centers at the coordinates
$(\frac12(i+j)\delta,h_{j-i})$ where $h_{d}<1$ is a suitable rational
(of small size) such that
$h_{d}^2+\frac14(d\delta)^2<1$ and $h_{d}^2+\frac14(d\delta+\delta)^2>1$.
Note that $D'_{i,j}$ intersects exactly $D_i,D_{i+1},\dots,D_j$,
and so the intersection graph of $\ca S\cup\ca N$
represents consecutive neighbourhoods of order~$n$.
\end{proof}
}

It is worthwhile to notice that for each of the classes listed 
in Corollary~\ref{cor:allhardness},
the $k$-clique and $k$-independent set problems are all easily FPT,
and yet \FO model checking is not.

\medskip
Finally, we return to the weaker \EFO model checking problem.
In fact, this problem can be treated ``the same'' as the aforementioned
parameterized induced subgraph isomorphism problem%
\shortlong{:
precisely, one of them admits an FPT algorithm on any given (unlabelled)
graph class if and only if the other does so.
\stepcounter{thm}
}{, which is a folklore result whose short proof we include for the sake of
completeness:
\begin{prop}
On any class $\ca G$ of simple unlabelled graphs,
the parameterized problems of induced subgraph isomorphism and of\/ \EFO model
checking are equivalent regarding FPT.
Precisely, one of them admits an FPT algorithm on $\ca G$ if and only if the
other does so.
\end{prop}
\begin{proof}
In one direction, given a graph $H$, $|V(H)|=k$, we straightforwardly construct a
quantifier-free \FO formula $\phi_H(x_1,\dots,x_k)$ such that
$G\models\phi_H(x_1,\dots,x_k)$ iff $G[x_1,\dots,x_k]$ is isomorphic to~$H$,
and $|\phi_H|$ is bounded in~$k$.
Then $\exists x_1,\dots,x_k\> \phi_H(x_1,\dots,x_k)$ is an \EFO sentence
solving the $H$-induced subgraph isomorphism problem on~$\ca G$.

In the other direction, assume an \EFO formula 
$\psi\equiv\exists x_1,\dots,x_k\>\psi_1(x_1,\dots,x_k)$
where $\psi_1$ is quantifier-free.
For a fixed vertex set $V=\{v_1,\dots,v_k\}$ (note; some vertices in this
list might be identical), let $\ca H_\psi$ denote the finite
set of all simple graphs on~$V$ such that $\psi_1(v_1,\dots,v_k)$ holds true for them.
Then the $\psi$-model checking problem on a graph $G$ reduces to checking
whether, for some $H\in\ca H_\psi$, the pair $\langle G,H\rangle$ 
is a {\sc Yes} instance of induced subgraph isomorphism.
Since $|\ca H_\psi|$ is bounded in $k\leq|\psi|$, the result follows.
\end{proof}
}

The hardness construction in the proof of Theorem~\ref{thm:allhardness} 
can be turned into \EFO, but only if vertex labels are allowed%
\shortlong{}{ (notice that in the proof, we introduced the universal quantifier
only when we had to remove the labels)}.
Though, we can modify some of the constructions from
Corollary~\ref{cor:allhardness} to capture also \EFO without labels.

\begin{prop}\apxmark\label{pro:EFOhardness}
The \EFO model checking problem is $W[1]$-hard with respect to the formula size,
for both the following unlabelled geometric graph classes:
\\a) circle graphs,
\\b) unit-box graphs with a representation contained within a square of side
length $3$.
\end{prop}

\shortlong{}{%
\begin{proof}
In the proof we carefully combine the respective constructions
from Corollary~\ref{cor:allhardness} with the first part of the proof of
Theorem~\ref{thm:allhardness}, so that universal quantifiers are avoided
-- this way we get the interpretation \eqref{eq:E-h-interpret} $I=(\nu,\psi)$
(labelled) which is actually \EFO.
Recall that $I$ is capable of interpreting any simple graph
$H$ in a suitable graph $G_H$ constructed in the considered class
in polynomial time, that is, $H\simeq I(G_H)$.

Then, in each of the considered cases, we will show an ad hoc modification of the
construction (see below) with the benefit of removing the colour labelling.
Before giving details of the modifications,
we show how the proof of $W[1]$-hardness is to be finished.

Consider the \EFO formula 
$$\gamma_k\>\equiv\> \exists x_1,\dots,x_k \left[
 	\bigwedge_{1\leq i\leq k} \nu(x_i) \wedge \!\!
 	\bigwedge_{1\leq i<j\leq k} \!\! x_i\not=x_j \wedge \!\!
 	\bigwedge_{1\leq i<j\leq k} \!\! \big(
         \psi(x_i,x_j)\vee \psi(x_j,x_i)  \big)\right]
;$$
 by the assumed interpretation, $G_H\models\gamma_k$ if and only if $H$ contains $k$
 vertices forming a clique, where the latter is a $W[1]$-hard problem with respect to~$k$.
 Since $|\gamma_k|$ is bounded in~$k$, this implies that
 the \EFO model checking instance $G_H\models\gamma_k$ 
 is also $W[1]$-hard with respect to $|\gamma_k|$, where $G_H$ is restricted
 to the considered graph class.

\smallskip
It remains to provide the ad hoc modified constructions
and the corresponding modifications of the formulas $\nu,\psi$ in~$\gamma_k$.

a) We turn the permutation witness (of consecutive neighbourhoods)
from Corollary \ref{cor:allhardness}b) into a circle representation
by joining the two parallel lines into one circle.
We then observe that no odd cycle $C_{2a+1}$ for $a\geq2$ is a permutation
graph since it does not have a transitive orientation, but
every odd cycle has a straightforward overlap representation.
Hence, if one wants to label a chord of a circle representation,
it is possible to do so by adding an adjacent small subrepresentation
of an odd cycle.

Namely, let $\ca D$ be the labelled circle representation of $G_H$
constructed for given~$H$ in the proof of Theorem~\ref{thm:allhardness}.
For each chord $a$ of $\ca D$ which has received label `$blue$',
we add a fresh copy of
(the representation of) $C_5$ with one vertex adjacent to~$a$.
We analogously add an adjacent copy of $C_7$ for every `$red$' chord 
and of $C_9$ for every `$green$' chord.
Let $\ca D'$ denote the new (unlabelled) circle representation
and $G_H'$ its intersection graph.
Since $G_H$ is actually a permutation graph by
Corollary~\ref{cor:allhardness}b),
the only induced $C_5,C_7,C_9$ in $G_H'$ are those later added ones.
Consequently, it is a routine task to express the predicate
$blue(x)$ in \EFO as `there exist vertices inducing $C_5$, and one is
adjacent to~$x$', and likewise for $red(x)$ and $green(x)$.
In this way, we get from $\gamma_k$ an \EFO formula $\gamma_k'$
such that $G_H\models\gamma_k$ if and only if $G_H'\models\gamma_k'$.

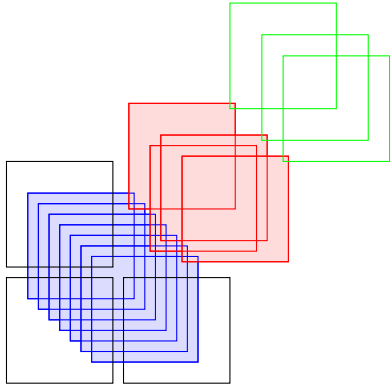
\begin{figure}[t]
$$
\begin{tikzpicture}[scale=1.4]
\tikzstyle{every path}=[draw,color=blue,
			 fill={rgb,255:red,220; green,220; blue,255}];
\def\bluebox{\foreach \x in {1,...,7} {
	\draw (0.1*\x,1-0.1*\x) rectangle (1+0.1*\x,2-0.1*\x) ;
}}\bluebox
\tikzstyle{every path}=[draw,color=red,
			fill={rgb,255:red,255; green,220; blue,220}];
\def\redbox{\draw (1.05,1.75) rectangle +(1,1) ;
	\draw (1.25,1.35) rectangle +(1,1) ;
	\draw (1.35,1.45) rectangle +(1,1) ;
	\draw (1.55,1.25) rectangle +(1,1) ;
}\redbox
\tikzstyle{every path}=[draw,color=blue,fill=none];
\bluebox
\tikzstyle{every path}=[draw,color=red,fill=none];
\redbox
\tikzstyle{every path}=[draw,color=black, fill=none];
	\draw (-0.1,0.1) rectangle +(1,1) ;
	\draw (-0.1,1.2) rectangle +(1,1) ;
	\draw (1.0,0.1) rectangle +(1,1) ;
\tikzstyle{every path}=[draw,color=green, fill=none];
	\draw (2.0,2.7) rectangle +(1,1) ;
	\draw (2.3,2.4) rectangle +(1,1) ;
	\draw (2.5,2.2) rectangle +(1,1) ;
\end{tikzpicture}
$$
\caption{Replacing explicit labels in the hardness construction
	of Proposition~\ref{pro:EFOhardness}b) -- adding the three black
	and several green boxes to the illustration in
	Figure~\ref{fig:construct-consecutiv} right.}
\label{fig:construct-consecutiv-EFO}
\end{figure}

b) This time we are not able to add ``local markers'' as in a),
since all boxes need to be of the same size.
Instead, we add just several new boxes to the whole
unit-box representation $\ca B$ from Corollary \ref{cor:allhardness}c).
See Figure~\ref{fig:construct-consecutiv-EFO};
the three black boxes are added to intersect precisely all the original blue
boxes, and one intersecting green box is added to every red box which,
in the proof of Theorem~\ref{thm:allhardness}, represents the successor
relation on blue boxes
(that is, which has received also label `$green$').

As one can easily check from the picture, we can now express the predicate
$blue(x)$ using \EFO as `there exist four independent neighbours of~$x$'
(this property is false for every other box type here).
Similarly, $red(x)$ can be expressed as `there exists a blue box
adjacent to $x$ and a blue box not adjacent to~$x$'.
Finally, $green(x)$ should be true for those red boxes which have a green
neighbour box, where a green box is characterised as having a blue and a red
{\em non}-neighbour.

The proof is then finished in the same way as in case a).
\end{proof}
}

One complexity question that remains open after Proposition~\ref{pro:EFOhardness} is
about \EFO on unlabelled permutation graphs
(for labelled ones, this is $W[1]$-hard by the remark after Corollary~\ref{cor:allhardness}).
While induced subgraph isomorphism is generally NP-hard
on permutation graphs by~\cite{HEGGERNES2015252},
we are not aware of results on the parameterized version,
and we currently have no plausible conjecture about its parameterized
complexity.

\section{Polygonal Visibility Graphs}\label{sec:visibility}

\subsection{Definitions}

Given a polygon $W$ in the plane, two vertices $p_i$ and $p_j$ of $W$ are said to be 
\emph{mutually visible} if the line segment $p_ip_j$ does not intersect the exterior of $W$.
The {\em visibility graph} $G$ of $W$ is defined to have vertices $v_i$
corresponding to each vertex $p_i$ of $W$,
and edge $(v_i,v_j)$ if and only if $p_i$ and $p_j$ are mutually visible. 
\shortlong{}{\par
Visibility graphs have been studied for several subclasses of polygons, such as orthogonal polygons, 
spiral polygons etc \cite{ec-rvgsp-90, ghks-ggpr-96, eow-sgra-84}. 
}%
Our aim is to study the visibility graphs of some special established classes of polygons 
with respect to \FO model checking.

If there is an edge $e$ of the polygon $W$, such that for any point $p$ of $W$, there is a point on $e$
that sees $p$, then $W$ is called a \emph{weak visibility polygon}, and $e$ is called
a \emph{weak visibility edge} of $W$
(Figure \ref{figtypes}a) \cite{g-vap-07, gmpsv-crwvp-93}.
A vertex $v_i$ of $W$ is called a \emph{reflex vertex} if the interior angle of $W$ formed at 
$v_i$ by the two edges of $W$ incident to $v_i$ is more than $\pi$.
Otherwise, $v_i$ is called a \emph{convex vertex}.
If both of the end vertices of an edge of $W$ are convex vertices, then the edge is called a \emph{convex edge}.

If the boundary of $W$ consists only of an x-monotone polygonal arc
touching the x-axis at its two extreme points, and an edge contained in the x-axis joining the two points,
then it is called a \emph{terrain} (Figure \ref{figtypes}b) \cite{g-vap-07, se-cgta-2002}. 
All terrains are weak visibility polygons with respect to their edge 
that lies on the x-axis.
If all points of a $W$ are visible from a single vertex $v$ of the polygon, then $W$ is called a \emph{fan}
(Figure \ref{figtypes}c) \cite{g-vap-07, gsbg-cmcvgsp-07}.
If $W$ is a fan with respect to a convex vertex $v$, then $W$ is called a \emph{convex fan} \cite{o-agta-87}.
If $W$ is a convex fan with respect to a vertex $v$, then both of the edges of $W$ incident to $v$ are convex edges,
and $W$ is also a weak visibility polygon with respect to any of them.
\begin{figure}[t]
\centering
\includegraphics[width=0.9\textwidth]{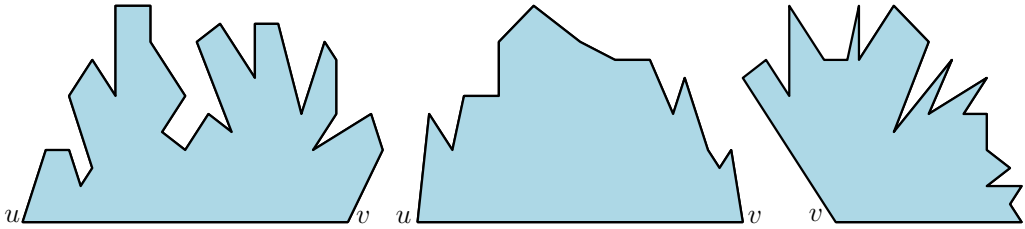}
\caption{From left to right: (a) a weak visibility polygon with respect to edge
	$uv$; (b) a terrain; (c) a convex fan visible from the vertex $v$.}
 \label{figtypes}
\end{figure}

In this section we identify some interesting tractable and hard cases of the
\FO model checking problem on these visibility classes.

\shortlong{}{%

\subsection{Hardness for terrain and convex fan visibility graphs}
}

We first argue that the \FO model checking problem of polygon visibility
graphs stays hard even when the polygon is a terrain and a convex fan.
Our approach is very similar to that in Theorem~\ref{thm:allhardness} above,
that is, we show that a given \FO model checking instance of general graphs
can be interpreted in another instance of the visibility graph of a
specially constructed polygon which is a terrain and a convex fan at the same time.
\shortlong{}{%
However, since polygon visibility graphs are in general not closed on
induced subgraphs and duplication of vertices,
we have to reformulate all the arguments from scratch.
}

\begin{thm}\apxmark\label{thm:TerFan-hard}
The \FO model checking problem of unlabelled polygon visibility graphs
(given alongside with the representing polygon) is 
$AW[*]$-complete with respect to the formula size,
even when the polygon is a terrain and a convex fan at the same time.
\end{thm}
\shortlong{%
\addtocounter{thm}{2}%
}{%

\begin{figure}[t]
\centering
\includegraphics[width=1.0\textwidth]{fig12}
\caption{The constructed polygon for a graph with $n=5$ vertices and edges $(v_1, v_3)$, $(v_1, v_5)$, $(v_2, v_4)$, $(v_2, v_5)$, $(v_3, v_4)$ and $(v_3, v_5)$. 
The purple vertex $u$ sees all vertices of $W$, while the brown vertices $v$ and $v'$ see only $u$, $u'$, 
all the blue vertices and each other. Each blue vertex $p_i$, $ 1 \leq i \leq 5$ represents the vertex $v_i$ of the graph. 
 }
 \label{fighard}
\end{figure}

\begin{proof}
 Consider a given graph $H$ with $n$ vertices and $m$ edges.
 We construct our polygon $W$ as follows (see Figure~\ref{fighard}):
 Consider an increasing, convex curve $C_1$ with respect to the x-axis.
 We mark $n+3$ points $p_0, p_1, p_2, \ldots, p_n, p_{n+1}$ and $w$ on $C_1$ from left to right.
 Each of the points will later be a vertex of the polygon, and $p_i$, $1 \leq i \leq n$
 will represent the vertices $v_i$ of the given graph~$H$.

 From $w$ onwards, we consider a decreasing convex curve $C_2\ni w$.
 For each ray $\overrightarrow {p_{i-1}p_{i}}$, $1 \leq i \leq n+1$, denote the point of intersection 
 of $C_2$ and $\overrightarrow {p_{i-1}p_{i}}$ by $q_{i}$.
 In the arc of $C_2$ between $q_i$ and $q_{i+1}$, we arbitrarily choose $n$
 pairwise disjoint subarcs $D_{i,j}$, $1\leq j\leq n$, of positive length.
 Now, for each edge $e_k = (v_i,v_j)\in E(H)$, $i < j$,
 we choose a point $s^1_k\in C_2$ arbitrarily in the interior of $D_{i,j}$.
 From a point slightly above $p_j$ on $C_1$, we start a ray that intersects $C_2$ at $s^1_k$.
 Now we mark a second point $s^2_k$ on this ray a tiny distance to the right of $s^1_k$
 (notice that $s^2_k$ is slightly above $C_2$).
 Finally, we drop a vertical ray downward from $s^2_k$ to intersect $C_2$ at a third point $s^3_k$.
 Note that these distances should be so small that also $s^3_k$ belongs to $D_{i,j}$.
 This ensures that among the $p_i$'s, $s^2_k$ sees exactly all points $p_i, p_{i+1},\ldots, p_j$,
 and $s^2_k$ is visible from any point below itself to the left.

 To finish the construction of $W$,
 we mark a point $w'$ on $C_2$ to the right of all the points marked so far.
 We drop two vertical rays $\overrightarrow{r_1}$ and $\overrightarrow{r_2}$ downward from $p_0$ and $w'$ respectively.
 We consider a point slightly above $p_{n+1}$ on $C_1$ and draw the lower tangent $\overrightarrow{r_3}$ from it to $C_2$. 
 Denote the point of intersection of $\overrightarrow{r_2}$ and $\overrightarrow{r_3}$ as $u'$.
 Intersect $\overrightarrow{wp_{n+1}}$ with $\overrightarrow{r_1}$, say, at point $x$. 
 Intersect $\overrightarrow{r_1}$ and $\overrightarrow{r_3}$ with a horizontal line below both $x$ and $u'$.
 Denote the point of intersection of $\overrightarrow{r_1}$ and the horizontal line as $u$.
 Mark a point $v$ slightly to the left of the intersection of the horizontal line and $\overrightarrow{r_3}$,
 so that $v$ cannot see any point on $C_2$. Mark another point $v'$ vertically
 slightly above $v$ on $\overrightarrow{r_3}$.

Now we draw the polygon by starting from $p_0$ and drawing the
polygonal boundary by connecting successive points embedded on
$C_1$ and then $C_2$ (including points $s^2_k$) from left to right. 
We complete $W$ by connecting with edges the remaining points in the
sequence $(w', u', v', v, u, p_0)$.
We summarise the properties of the resultant polygon $W$ and its visibility graph~$G$:
\begin{itemize}
\item
$W$ is a terrain with respect to $uv$ and a convex fan with respect to $v$,
\item
no two points among $\{p_0,p_1,\ldots,p_{n+1}\}$ see each other except the consecutive pairs,
\item
for every $1\leq k\leq m$, the points $s^1_k,s^3_k$ see a consecutive
strip of $\{p_0,p_1,\ldots,p_{n+1}\}$ including $p_{n+1}$,
while $s^2_k$ can see $p_i, p_{i+1},\ldots, p_j$ but neither
$p_{i-1},p_{j+1}$ nor $p_{n+1}$, and
\item
the vertices $v$ and $v'$ are true twins in $G$ -- they see the same neighbourhood
which (except~$v,v'$) is $\{p_0,p_1,\ldots,p_{n+1},u,u'\}$,
and there is no other twin pair in~$G$.
\end{itemize}

\begin{claim}
Construction of the polygon $W$ can be finished in polynomial time.
\end{claim}

Since the constructed visibility graph $G$ is clearly of polynomial size with
respect to given~$H$, we only need to show that we can finish our construction
of $W$ with rational coordinates of sufficiently small size.
To argue this, we choose suitable curves $C_1,C_2$ such as quadratic
functions $y=(x+c)^2$ for appropriate values of~$c$.
We pick $p_0, p_1, \ldots, p_n, p_{n+1}$ and $w$ as grid points on $C_1$.
The positions of $q_1,\dots,q_{n+1}$ are computed only approximately
(they are not vertices of $W$ anyway), and then we choose the subarcs
$D_{i,j}$ with suitable (small) rational coordinates.
Subsequent choices of $s^1_k,s^2_k,s^3_k$ can also be done with
rational coordinates of small size, for $1\leq k\leq m$.
The remaining vertices of $W$ follow easily.

\begin{claim}
There exists a pair of \FO formulas $I=(\nu,\psi)$ (an \FO interpretation)
such that, for any given graph $H$, the resultant visibility graph~$G$ (as
above) satisfies $H\simeq I(G)$.
\end{claim}

We stress that the graph $G$ we have constructed is unlabelled, but for
clarity we will refer to the vertex colours introduced in
Figure~\ref{fighard}.
Recall, from the proof of Theorem~\ref{thm:allhardness}, the formula 
$twin(x,y) \equiv edge(x,y)\wedge \forall z [(z\neq x \wedge z\neq y)
 \rightarrow (edge(x,z)\leftrightarrow edge(y,z))]$
asserting that $x,y$ are true twins.
Since $v,v'$ are the only twins in $G$, we may match either of them with
the formula
$$ brown(x)\>\equiv\> \exists t\, x\not=t\wedge twin(x,t). $$
Subsequently, the vertices $p_0,p_1,\ldots,p_{n+1}$ are precisely those
matched by the formula
$$ blue(x)\>\equiv\> \exists z\big[
	brown(z)\wedge edge(x,z)\wedge \exists t
		(edge(t,z)\wedge\neg edge(t,x))
\big] $$
since, among all the neighbours of $v$, the vertices $u,u'$ see all the
other neighbours of~$v$.

The vertex set of $H$ (in the interpretation $I$) can hence be defined using
$$ \nu(x)\>\equiv\> blue(x)\wedge \exists z,z' \big(
	z\not=z'\wedge blue(z)\wedge blue(z')\wedge edge(x,z)\wedge
	edge(x,z') \big), $$	
which excludes $p_0$ and $p_{n+1}$ from the list of blue points.
Recall that every edge $e_k = (v_i,v_j)\in E(H)$, $i<j$, is represented
by the red vertex $s^2_k$ which sees precisely $p_i,p_{i+1},\dots,p_j$
among the blue points.
Our aim, in the formula $\psi(x,y)$ of $I$, is to specify that $x=p_i$ and
$y=p_j$ (or vice versa), and this can be done by referring to the unique blue
neighbours $p_{i-1}$ and $p_{j+1}$ of $x$ and $y$, respectively, which do
not see $s^2_k$.
(This part is the reason why we use blue $p_0,p_{n+1}$ in our construction.)
We write down this as follows
\begin{align*}
\psi(x,y)\>\equiv\>& blue(x)\wedge blue(y)\wedge\>
	\exists z,t,t' \big[\,
		blue(t)\wedge blue(t')\wedge edge(x,t)\wedge edge(y,t')
\\
	&\wedge\neg blue(z)\wedge edge(x,z)\wedge edge(y,z)
			\wedge\neg edge(t,z)\wedge\neg edge(t',z)
	\big]
.\end{align*}
Then $G\models\psi(v_i,v_j)$ if, and only if, $(v_i,v_j)\in E(H)$.

The rest of the proof is as in Theorem~\ref{thm:allhardness}.
\end{proof}
}

\shortlong{}{%

\subsection{Visibility graphs of weak visibility polygons of convex edges} 
}

\shortlong{Second,}{In this section}
we prove that \FO model checking of the visibility graph of a given weak visibility 
polygon of a convex edge is FPT when additionally parameterized 
by the number of reflex vertices.
We remark that, for example, the 
independent set problem is NP-hard on polygonal visibility graphs~\cite{s-hpip-89}, 
but Ghosh et al.~\cite{gmpsv-crwvp-93} showed that the maximum independent set 
of the visibility graph of a given weak visibility polygon of a convex edge, is
computable in quadratic time.
In Theorem~\ref{thm:TerFan-hard}, we have seen that the latter result does
not generalise to arbitrary \FO properties, since \FO model checking remains
hard even for a very special subcase of weak visibility polygons.
So, an additional parameterization 
\shortlong{}{in the next theorem }is necessary.

\begin{thm}\apxmark\label{thm:VIStractable}
Let $W$ be a given polygon weakly visible from one of its convex edges, with $k$ reflex vertices,
and let $G$ be the visibility graph of~$W$.
Then \FO model checking of $G$
 is FPT with respect to the parameters $k$ and the formula size.
\end{thm}

\shortlong{%
While we cannot fit the whole algorithm in the short paper,
we at least give an informal overview of how the algorithm works.
}{%
Before diving into the technical details of the rather long proof, we first provide
a brief informal summary of the coming steps.
}%
As in the previous intersection graph cases, our aim is to construct, from
given~$W$, a poset $\ca P$ such that the width of $\ca P$ is bounded by a
function of~$k$ and that we have an \FO interpretation of the visibility
graph of $W$ in this~$\ca P$.

Let $W$ be weakly visible from its convex edge $uv$, and denote by
$C_{uv}$ the clockwise sequence of the vertices of $W$ from $u$ to~$v$.
The subsequence of $C_{uv}$ between two reflex~vertices $v_a$ and $v_b$, 
such that all vertices in it are convex, is called an \emph{ear} of $W$.
The length of~this sequence can be $0$ as well.
Additionally, the first (last) ear of $W$ is defined as the subsequence between $u$ and
the first reflex vertex of $C_{uv}$ (between the last reflex vertex and~$v$, respectively).
We have got $k+1$ ears in~$W$.
With a slight abuse of terminology at $u,v$, we may simply say that an ear 
is a sequence of convex vertices between two reflex vertices.

The crucial idea of our construction of the poset $\ca P$
(which contains all vertices of $W$, in particular) is that
the visibility edges between the internal (convex) vertices of the ears are
nicely structured:
withing one ear $E_a$, they form a clique, and between two ears $E_a,E_b$,
the visibility edges exhibit a ``shifting pattern'' not much different from
the left and right ends of intervals in a proper interval representation
(cf.~Lemma~\ref{INTtoposet}).
Consequently, we may ``encode'' all the edges between $E_a$ and $E_b$ with
help of an extra subposet of $\ca P$ of fixed~width, and since we have got only
$k+1$ ears, this together gives a poset of width bounded in~$k$.

The last step concerns visibility edges incident with one of the $k$ reflex
vertices or $u,v$.
These can be easily encoded in $\ca P$ with only $2(k+2)$ additional labels,
without any assumption on the structure of $\ca P$:
for each reflex vertex $x$ of $C_{uv}$, or $x\in\{u,v\}$, 
we assign one new label $L^0_x$ to $x$
itself and another new label $L^1_x$ to all the neighbours of $x$.
Altogether, we can efficiently construct an \FO interpretation of $G$ in
$\ca P$ such that the formulas depend only on~$k$.
Then we may finish by Theorem~\ref{thm:posetFPT}.

\shortlong{}{%
 \proof[Proof of Theorem~\ref{thm:VIStractable}]
Throughout the proof (rest of the section) we will implicitly assume
a polygon $W$ which is weakly visible from its edge $uv$, 
where $uv$ is a convex edge of $W$, and the clockwise boundary from $u$ to $v$, 
denoted as $C_{uv}$, contains all the other edges of $W$. 
We also recall that $C_{uv}$ consists of $k+1$ ears.
 Let $G=(V,E)$ be the visibility graph of $W$. 

We need more terminology and some specialised claims.

For two elements $p,q$ of a poset, we say that {\em$q$ covers $p$}
if $p\sqsubseteq q$ and there is no poset element $r$ such that
$p\sqsubseteq r\sqsubseteq q$ and $p\not=r\not=q$.

 A vertex $z$ of $w$ is said to \emph{block} two vertices $v_i$ and $v_j$ of $W$ if the shortest path
 between $v_i$ and $v_j$ that does not intersect the exterior of $W$, takes a turn at $z$.
For two vertices $a$ and $b$ of $W$, when we say $a$ \emph{precedes} $b$ or $b$ \emph{succeeds} $a$ on $C_{uv}$,
we mean that we encounter $a$ earlier than $b$ when we traverse $C_{uv}$ in the clockwise order, starting from $u$.
 
\begin{claim} \label{vischar}
 Let $E_a$ and $E_b$ be two ears of $W$ such that $E_a$ precedes $E_b$ on $C_{uv}$.
 Let $v_a$ and $v_b$ be any convex vertices of $E_a$ and $v_i$ and $v_j$ be
 any convex vertices of $E_b$, where $v_a$ precedes $v_b$ and $v_i$ precedes $v_j$ on $C_{uv}$.
 Then the following hold.
 If $v_a$ sees $v_i$, then $v_a$ also sees $v_j$. 
 Symmetrically, if $v_j$ sees $v_b$, then $v_j$ also sees $v_a$.
\end{claim}
\begin{proof}
 Suppose that $v_a$ does not see $v_j$. Then there must be a blocker of $v_a$ and $v_j$. Since $v_i$ and $v_j$ are 
 convex vertices of the same ear, the blocker cannot come from the polygonal boundary in between them. 
 Since the $v_av_i$ lies inside $W$, the blocker also cannot come from the clockwise polygonal boundary between $v_a$ and $v_i$.
 If the blocker comes from the clockwise polygonal boundary between $u$ and $v_a$ then $v_a$ cannot see any part $uv$, a contradiction.
 Similarly, the blocker cannot come from the clockwise polygonal boundary between $v_j$ and $v$ as well.
 So, $v_a$ must see $v_j$.
 The second claim follows from symmetrical arguments.
 \end{proof}

Now we describe our construction of the poset~$\ca P=(P,\leq^\PPP)$ where
$P$ includes the vertices $V$ of~$W$.
We start with a linear order $\leq^C$ on the vertex set $V$ defined as follows.
For two vertices $a$ and $b$ of $V$, we let $a \leq^ C b$ iff 
$a$ precedes $b$ in the clockwise order on $C_{uv}$ or $a=b$.
We give all elements of $V$ label `$green$' and, additionally,
give label `$black$' to those which are reflex vertices of $W$ and to~$u,v$.
Let $\leq^C$ be a subrelation of $\leq^\PPP$. We have:

\begin{claim}\label{oneear}
It can be expressed in \FO that two vertices of $C_{uv}$ belong to the same ear.
\end{claim}
\begin{proof}
We give the formula
$$ \beta_0 (x,y) \equiv green(x) \wedge green(y) \wedge
	x \leq^\PPP\! y \wedge\,   \forall z \left[
	( x \leq^\PPP\! z \wedge z \leq^\PPP\! y \wedge black(z) )
		 \rightarrow ( z = x \vee z = y)\right]
$$
and use its symmetric closure $\beta_0(x,y) \vee \beta_0 (y,x)$.
\end{proof}

\begin{figure}[th]
 \centering
 \includegraphics[width=0.8\textwidth]{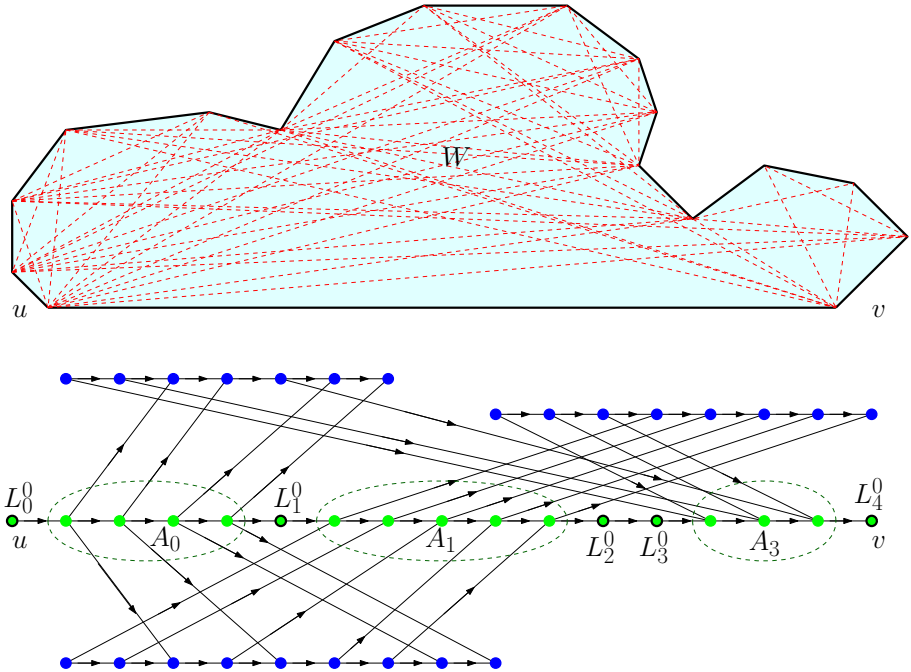}
\medskip
 \caption{An illustration of a weak visibility polygon $W$ and
	 its constructed poset, as in the proof of
	Theorem~\ref{thm:VIStractable}.
	Here, the sequence $C_{uv}$ (the green chain of the poset)
	consists of four ears $E_0,E_1,E_2,E_3$ with interiors $A_0,A_1,A_2,A_3$,
	where $A_2$ is empty (has no convex vertices).
	So, there are three blue chains (top to bottom)
	$B_{0,3},B_{1,3},B_{0,1}$ in the picture.
	The dashed lines in $W$ are the visibility edges of~$G$.
  }
  \label{figposet}
\end{figure}

Next, we number the ears of $C_{uv}$ as $E_0,E_1,\dots,E_k$ in the clockwise
order.
For every pair $0\leq a<b\leq k$,
we now describe a subposet of $\ca P$ which we will use to encode the edges
between the convex vertices of $E_a$ and~$E_b$.
Let $A_a$ and $A_b$ be the sets of convex vertices of $E_a$ and $E_b$, respectively,
and let $B_{a,b}$ denote a fresh disjoint copy of $(A_a\cup A_b)$.
For each $v_i\in A_a$ and its corresponding copy $v_i'\in B_{a,b}$,
we have $v_i\leq^\PPP v_i'$.
Analogously, for each $v_j\in A_b$ and its corresponding copy 
$v_j'\in B_{a,b}$, we have $v_j'\leq^\PPP v_j$
and, in fact, it holds that $v_i'$ covers $v_i$ and $v_j$ covers $v_j'$.
The whole set $B_{a,b}$ is made into a chain of $\ca P$
ordered such that, for any $v_i,v_j\in A_a\cup A_b$ and their 
corresponding copies $v_i',v_j'\in B_{a,b}$, we have
\begin{itemize}
\item if either $v_i,v_j\in A_a$ or $v_i,v_j\in A_b$,
 then $v_i'\leq^\PPP v_j'$ iff $v_i\leq^\PPP v_j$;
\item if (up to symmetry) $v_i\in A_a$ and $v_j\in A_b$,
 then $v_i'\leq^\PPP v_j'$ iff $v_i$ can see $v_j$ in~$W$.
\end{itemize}
We give all the elements of $B_{a,b}$, $0\leq a<b\leq k$, the label
`$blue$', and will refer to each such $B_{a,b}$ as to a {\em blue chain}.
See Figure~\ref{figposet}.

By Claim~\ref{vischar}, $\leq^\PPP$ forms a valid (sub)poset on
$V\cup B_{a,b}$.
Now we make $\ca P$ the union of the subposets considered so far
(green $V$ and the blue chains), with a transitive closure of~$\leq^\PPP$.
That is, $P=V\,\bigcup_{0\leq a<b\leq k}B_{a,b}$
and $\leq^\PPP$ restricted to each $V\cup B_{a,b}$ is as defined above.

\begin{claim}\label{twoears}
It can be expressed in \FO that two convex vertices $v_i\in E_a$ and $v_j\in
E_b$ see each other, i.e., they form an edge of~$G$.
\end{claim}
\begin{proof}
Assume, up to symmetry, $v_i\leq^\PPP v_j$ and $a\not=b$.
By the definition of $\leq^\PPP$ on $B_{a,b}$ we have that $v_i$ can see
$v_j$ if and only if there are copies $v_i',v_j'\in B_{a,b}$ such that
$v_i'\leq^\PPP v_j'$.
The latter, however, is not so simple to express since blue elements 
of $\ca P$ comparable with $v_i,v_j$ exist on other blue chains than
$B_{a,b}$, due to transitivity.
Moreover, $v_i'\leq^\PPP v_j'$ does not imply that  $v_i',v_j'$ belong to
the same blue chain, again, due to transitivity (``through'' some green vertex of~$V$).

Hence, we are going to express that
$v_i'$ covers $v_i$, $v_j$ covers $v_j'$, and
that $v_i'\leq^\PPP v_j'$ indeed belong to the same blue chain.
For the former, we give the following \FO formula
$$
cover(x,y)\>\equiv\> x\leq^\PPP\! y \wedge
	\forall z\left[ x\leq^\PPP z\leq^\PPP y \to
		(x=z\vee y=z)\right]
,$$
and for the latter assertion, we may write 
(implicitly assuming $blue(x)\wedge blue(y)$ as below)
$$
samechain(x,y)\>\equiv\> 
	\forall z\left[ (x\leq^\PPP z\leq^\PPP y
		\vee y\leq^\PPP z\leq^\PPP x) \to \neg green(z)\right]
.$$
Together, we formulate
$$
see(x,y)\>\equiv\> \exists z,t \big[ blue(z)\wedge blue(t)\wedge
	samechain(z,t)\wedge cover(x,z)\wedge z\leq^\PPP t
		\wedge cover(t,y) \big]
$$
and, with additional identification of convex vertices of the ears, we
finally get
$$
\beta_1(x,y)\>\equiv\> green(x)\wedge green(y)\wedge
	\neg black(x)\wedge\neg black(y)\wedge
		\big( see(x,y)\vee see(y,x) \big)
.$$

We claim that $\ca P\models\beta_1(v_i,v_j)$, if and only if
$v_i,v_j$ are convex vertices of distinct ears and they see each other.
In the backward direction, if $v_i,v_j$ see each other, then
$\ca P\models\beta_1(v_i,v_j)$ is witnessed by the choice
of $\{z,t\}=\{v_i',v_j'\}$ in $see(x,y)$.

On the other hand, assume $\ca P\models\beta_1(v_i,v_j)$.
Then $v_i,v_j$ are convex vertices of some ears $E_a\ni v_i$ and $E_b\ni v_j$
of $C_{uv}$, by the labels `$green$' and `$\neg black$'.
Up to symmetry, $\ca P\models see(v_i,v_j)$.
From $cover(v_i,z)$ we know that $z\in B_{a,b'}$ for some $b'$,
and from $cover(t,v_j)$ we get $t\in B_{a',b}$ for some $a'$.
By $samechain(z,t)$, it holds $a=a'$ and $b=b'$.
Consequently, by the definition of $\leq^\PPP$ on $V\cup B_{a,b}$
we get that $v_i$ sees~$v_j$ in~$W$.
\end{proof}

It remains to address the edges of $G$ which are incident with one or two 
reflex vertices of $W$ or $u$ or $v$.
Let $r_0=u, r_1,\dots,r_k, r_{k+1}=v$ be the clockwise order of
$u,v$ and the reflex vertices on $C_{uv}$.
We assign every $r_i$, $0\leq i\leq k+1$, in $\ca P$
a new label $L_i^0$, and then assign another new label $L_i^1$
to all the vertices of $V$ adjacent to~$r_i$.

\begin{claim}\label{ereflex}
Let $v_i$ be a reflex vertex or one of $u,v$, and $v_j\in V$.
It can be expressed in \FO that $v_i,v_j$ form an edge of~$G$.
\end{claim}
\begin{proof}
This is trivial (up to symmetry):
$$
\beta_2(x,y)\>\equiv\> black(x)\wedge
	\bigvee\nolimits_{0\leq i\leq k+1} \left(
		L_i^0(x)\wedge L_i^1(y) \right)
.\vspace*{-\baselineskip}$$
\end{proof}

We have constructed the poset $\ca P$ in polynomial time from the given
polygon~$W$, and the width of $\ca P$ is at most ${k+1\choose2}+1$
since we have created one new chain for each pair of distinct ears.
We finish the proof, by Theorem~\ref{thm:posetFPT}, if we provide
an \FO interpretation $I=(\nu,\psi)$ depending only on~$k$,
such that $G=I(\ca P)$;
$$ \nu(x)\>\equiv\> green(x), $$
$$ \psi(x,y)\>\equiv\> green(x)\wedge green(y)\wedge
	\big[ \beta_0(x,y) \vee \beta_0 (y,x)  \vee \beta_1(x,y) \vee \beta_1 (y,x)
	  \vee \beta_2(x,y) \vee \beta_2 (y,x) \big]
.$$

Validity of this interpretation follows from the fact that the edge set of
$G$ is a union of cliques on each of the ears and
of edges between convex vertices of distinct ears and of edges incident
with reflex vertices or $u$ or $v$,
and from Claims~\ref{oneear}, \ref{twoears}, \ref{ereflex}.
 \qed
}

\section{Conclusions}\label{sec:conclu}

We have identified several FP tractable cases of the \FO model checking problem of
geometric graphs, and complemented these by hardness results showing quite
strict limits of FP tractability on the studied classes.
Overall, this presents a nontrivial new contribution towards understanding
on which (hereditary) dense graph classes can \FO model checking be FPT.

All our tractability results rely on the \FO model checking algorithm
of~\cite{gajarskyetal15}, which is mainly of theoretical interest.
However, in some cases one can employ, in the same way, the
simple and practical \EFO model checking algorithm of~\cite{ghoo14}.
We would also like to mention the possibility of enhancing the result
of~\cite{gajarskyetal15} via interpreting posets in posets.
While this might seem impossible, we actually have one 
positive indication of such an enhancement.
It is known that interval graphs are $C_4$-free complements of comparability
graphs (i.e., of posets) -- the width of which is the maximum clique size of
the original interval graph.
Then, among $k$-fold proper interval graphs there are ones of unbounded
clique size, which have FPT \FO model checking by
Theorem~\ref{thm:INTtractable}.
This opens a promising possibility of an FP tractable subcase 
of \FO model checking of posets of unbounded width, for future research.

To complement previous general suggestions of future research,
we also list two concrete open problems which are directly related to our results.
We conjecture that \FO model checking is FPT
\begin{itemize}
\item for circle graphs additionally parameterized by the maximum clique size, and
\item for visibility graphs of weak visibility polygons
	additionally parameterized by the maximum independent set size.
\end{itemize}


\bibliography{FO-geom-IPEC-journal}

\end{document}